\documentclass{article}
\usepackage{iclr2027_conference,times}
\iclrfinalcopy
\usepackage{etoolbox}
\makeatletter
\patchcmd{\@maketitle}{Published as a conference paper at ICLR 2027}{Preprint}{}{\PackageError{arxiv-preprint}{Could not replace the conference header}{Check the title formatting.}}
\makeatother
\usepackage[T1]{fontenc}
\usepackage[utf8]{inputenc}
\usepackage{amsmath,amssymb,amsthm,mathtools,bm}
\usepackage{graphicx,booktabs,tabularx,multirow,xcolor}
\usepackage{microtype}
\usepackage[hidelinks]{hyperref}
\usepackage{url}
\graphicspath{{figures/}}
\newtheorem{theorem}{Theorem}
\newtheorem{proposition}[theorem]{Proposition}
\newtheorem{lemma}[theorem]{Lemma}

\theoremstyle{definition}

\newcommand{\E}{\mathbb E}
\newcommand{\Prob}{\mathbb P}
\newcommand{\R}{\mathbb R}
\newcommand{\ind}{\mathbf 1}
\newcommand{\ones}{\mathbf 1}

\newcommand{\Poi}{\operatorname{Pois}}

\hypersetup{pdfauthor={Rudra Chopra},pdfsubject={Research preprint}}

\AtBeginDocument{\lhead{Preprint}}

\title{Calibration Count Reuse:\\Validity Does Not Determine Efficiency}
\author{Rudra Chopra}
\hypersetup{pdftitle={Calibration Count Reuse: Validity Does Not Determine Efficiency}}
\begin{document}
\maketitle
\begin{abstract}
Calibration count reuse raises separate validity and efficiency questions. We give a validity criterion for general count-dependent nonconformity scores: transferring one count from another class to the scored class must not improve its conformity. A leave-self-out full conformal reference proves the criterion without requiring normalization or preservation of same-class score order. For a common separable transformation, universal exchangeable validity is equivalent to being nondecreasing in the count, provided $K\alpha\geq1$; normalized multiplicative weights obey the complementary nonincreasing condition. Additive penalties are covered under the stated information restrictions. Efficiency has no parallel ordering: two iid constructions make the same valid rule improve or worsen expected size at unchanged coverage. An expanded 55-rule study finds no resolved advantage from selected live-count rules over uniform weights. Image studies identify undercoverage under iid resampling, including at numerical convergence. Separately, execution of the released Conf-OT pipeline on its DTD and Aircraft benchmark subsets produces near-nominal median coverage under fixed stratified counts. The native results are reported separately from the iid analyses, without treating a benchmark observation as a universal guarantee. The findings separate validity, classifier confidence, numerical convergence, and population-specific efficiency.
\end{abstract}
\section{Introduction}
Class frequency adjustment raises two different questions for conformal prediction. Can final calibration labels safely influence a score that is calibrated on those same labels? If they can, does this improve prediction set efficiency? A validity theorem does not answer the second question, and an average set-size trend does not establish a universal optimum.

A calibration observation contributes to its own label count, whereas a query does not. We characterize a direction of score change under which this asymmetry is conservative. The key operation transfers one count from another class to the class being scored. If that operation cannot make the score more conforming, comparison with a leave-self-out full conformal reference proves validity. This argument covers both normalized multiplicative weights and additive penalties, without requiring a shared normalizer or preservation of same-class score order.

For a common separable transformation $S_i(h;c)=\phi(V_i(h),c_h)$, universal exchangeable validity is equivalent to being nondecreasing in the count, provided $K\alpha\geq1$. No monotonicity in the feature mark $V_i(h)$ is needed. For normalized probabilities, the corresponding exact boundary is nonincreasing multiplicative weights. The complete theorem and both converses are in Section~\ref{sec:boundary}. These are universal statements over permitted bases, laws, and sample sizes, not characterizations for every fixed model or iid population.

Efficiency has no such ordering. Theorem~\ref{thm:efficiency} gives two iid laws where the same valid count rule respectively removes and adds false labels at unchanged coverage. Our 55-rule study consequently compares live-count rules with equally informed frozen-weight controls. Empirical prior transport, motivated by Conf\mbox{-}OT \citep{silva2025}, supplies a contrasting invalid use. A worked rank flip precedes the theory, followed by exact failure constructions, the frozen image study, and separately reported retrospective analyses.

\section{A two class rank flip}\label{sec:toy}
Consider nine calibration observations: five of class $A$ and four of class $B$. The target is a class $B$ observation. Every base score predicts the correct class. The positive class weights are $f(c)=c+1$, giving calibration weights $(6,5)$. At nominal $90\%$, the rank is $k=\lceil10\cdot.9\rceil=9$: a candidate is rejected if all nine calibration true probabilities are strictly higher.

\begin{figure}[ht]
\centering
\setlength{\tabcolsep}{5pt}
\begin{tabular}{lccc}
\toprule
Rows and true class & Base $(A,B)$ & Counts $(5,4)$ & Add candidate $B$\\
 & & Weights $(6,5)$ & Weights $(6,6)$\\
\midrule
Five calibration $A$ rows & $(5/2,1)$ & $3/4$ & $5/7$\\
Four calibration $B$ rows & $(1,4)$ & $10/13$ & $4/5$\\
Target $B$ row & $(1,3)$ & $5/7$ & $3/4$\\
\midrule
Calibration scores above target & & $9$ & $4$\\
Decision for candidate $B$ & & \textbf{Exclude} & \textbf{Include}\\
\bottomrule
\end{tabular}
\caption{\textbf{The mechanism, in exact fractions.} Displayed probabilities are for each row's true label. The features, labels, and base scores never change; only the candidate's count is added. Augmentation moves the target above all five class $A$ scores while preserving its order relative to class $B$ scores. This is a samplewise rank flip, not a claim of zero marginal coverage for this two class law. In particular, it does not invoke the later $K\alpha\geq1$ converse.}
\label{fig:mechanism}
\end{figure}

The ordinary procedure excludes the true label even though the target's unadjusted probability is $3/4$. Full augmentation recomputes \emph{both} calibration scores and the candidate score using $(6,6)$ and includes it. Changing only the target probability would be a different, unjustified procedure. The example isolates the comparison that fails; the theorem below determines when the same direction holds for every possible sample.

\section{Setup and a validity criterion beyond multiplicative scores}\label{sec:setup}
\paragraph{Permitted information.} There are $n\geq1$ calibration pairs, $M\geq1$ query rows, and a fixed catalog $[K]$. Conditional on independently fitted information and the other query features, calibration plus the chosen target must remain exchangeable. Feature marks $V_i(h)$ may use the pooled features but not final calibration labels; they permute with the rows. Fixed measurable score functions $\Psi_h$ produce finite nonconformity scores
\[
S_i(h;c)=\Psi_h(V_i(h),c),\qquad c_h=\sum_{i=1}^n\ind\{Y_i=h\}.
\]
Smaller scores are more conforming. Hyperparameters may be fitted independently, not chosen from final calibration outcomes. One important special case is $S_i(h;c)=-p_i(h;c)$ with positive feature-only $A$ and fixed positive functions $f_h$:
\begin{equation}\label{eq:generalweight}
p_i(h;c)=\frac{A_{ih}f_h(c_h)}{\sum_j A_{ij}f_j(c_j)}.
\end{equation}

\paragraph{Inclusive ranks and two references.} The ordinary set $C_E$ uses
\begin{equation}\label{eq:rank}
k_\alpha=\lceil(n+1)(1-\alpha)\rceil,\quad
b_h=\#\{i\leq n:S_i(Y_i;c)<S_*(h;c)\},\quad h\in C_E\ \Longleftrightarrow\ b_h<k_\alpha.
\end{equation}
Exact ties remain included; $k_\alpha=n+1$ includes every label. The all-count reference $C_A$ replaces $c$ by $c+e_h$ in \emph{every} score for candidate $h$. Our broader argument uses a different reference $C_{\rm loo}$: score calibration row $i$ using $c+e_h-e_{Y_i}$ and the query using $c$. At the true hypothesis, every row is scored using the full label count minus its own label. Both reference score vectors are therefore equivariant, and classical full conformal ranking \citep{shafer2008} gives
\begin{equation}\label{eq:reference}
\Pr\{Y_*\in C_A\},\ \Pr\{Y_*\in C_{\rm loo}\}\ \geq\ k_\alpha/(n+1)\ \geq\ 1-\alpha.
\end{equation}
Neither reference fits a new feature representation. $C_A$ and $C_{\rm loo}$ need not be the same set.

\subsection{A transfer criterion and two sharp boundaries}\label{sec:boundary}
\begin{theorem}[Count-transfer validity and sharp specializations]\label{thm:sharpcommon}
Under the information and sampling conditions above, suppose that for every permitted mark $v$, count vector $c$, and $j\ne h$ with $c_h\geq1$,
\begin{equation}\label{eq:transfer}
\Psi_j(v,c+e_j-e_h)\geq\Psi_j(v,c).
\end{equation}
Then $C_{\rm loo}\subseteq C_E$ for every sample and $C_E$ has marginal coverage at least $1-\alpha$. Reversing all inequalities reverses containment, but does not itself imply nominal undercoverage.

For a common separable map $\Psi_h(v,c)=\phi(v,c_h)$ fixed across classes and sample sizes, fix $K\alpha\geq1$. Universal validity over all sample sizes, admissible exchangeable laws, and permitted feature marks holds if and only if $\phi(v,m+1)\geq\phi(v,m)$ for every $v,m$. Its sufficient direction permits different maps $\phi_h$ and imposes no monotonicity in $v$.

For the normalized family \eqref{eq:generalweight}, the corresponding common-function characterization is that $f$ must be nonincreasing. Its sufficient direction permits different nonincreasing $f_h$. Every prohibited step in either common-function characterization admits an exchangeable zero-coverage law at a specified sample size.
\end{theorem}
\begin{proof}
Fix candidate $h$ and put $D=c+e_h$. In its leave-self-out reference, calibration row $i$ with label $j$ has score $R_i=\Psi_j(V_i(j),D-e_j)$; the query has $R_*=\Psi_h(V_*(h),c)$, exactly its ordinary score. If $j=h$, the ordinary and reference calibration scores coincide. If $j\ne h$, apply \eqref{eq:transfer} at $D-e_j$ to obtain
\[
S_i(j;c)=\Psi_j(V_i(j),D-e_h)\geq\Psi_j(V_i(j),D-e_j)=R_i.
\]
Thus the ordinary number of calibration scores strictly below the query is no larger than the reference number. This proves containment with arbitrary ties, and \eqref{eq:reference} gives coverage. The reverse direction is identical with reversed inequalities.

A separable map nondecreasing in its own count satisfies \eqref{eq:transfer}. Conversely, suppose $\phi(v_0,m)<\phi(v_0,m-1)$ for some $m\geq1$. Uniformly permute $m$ observations from each class and hold one out, so $n=Km-1$. Let every true-label feature mark equal $v_0$, a permitted feature-only constant base. The target is scored at $m-1$; all $(K-1)m$ other-class calibration rows are scored at $m$ and are strictly more conforming. Since $\lceil Km(1-\alpha)\rceil\leq(K-1)m$, every target is excluded.

For normalized probabilities, adding to the scored class count lowers its nonincreasing weight, while removing a count from another class raises that competing weight. Its probability cannot increase, so $S=-p$ satisfies \eqref{eq:transfer}. For necessity suppose $f(m)>f(m-1)$ and use the same balanced labels with the feature-only base $A_{ij}=R^{\ind\{Y_i=j\}}$, $R>1$. With $q=f(m-1)/f(m)<1$, the target true probability is $Rq/(Rq+K-1)$ and every other-class true probability is $R/(R+K-2+q)$. The latter is larger because
\[
(Rq+K-1)-q(R+K-2+q)=(1-q)(K-1+q)>0.
\]
The same count inequality excludes every target. This witness has perfect native classification. Taking $R$ large and then adding sufficiently small continuous feature noise preserves strict comparisons and achieves any prescribed confidence below one; Appendix~A gives the perturbation detail.
\end{proof}

\paragraph{Quantifiers.} The converses fix the common map, $K$, and $\alpha$, then range over sample sizes, laws, and permitted bases. They are not iff statements at each fixed $n$, iid characterizations, or heterogeneous-map necessity results. A sufficient map can depend on known $n$ or independent fitting information; the common-map converses do not automatically extend to arbitrary sample-size-dependent families. The $K\alpha\geq1$ restriction belongs only to the converses.

\subsection{Additive scores and why the reference matters}\label{sec:additive}
A separable additive score $S_i(h;c)=s_i(h)+g_h(c_h)$ is valid whenever $g_h$ is nondecreasing. More generally, for fixed $\lambda\geq0$, $a\geq0$, and feature-only penalties $w_i(h)\geq0$,
\begin{equation}\label{eq:additive}
S_i(h;c)=s_i(h)+\lambda w_i(h)\frac{c_h+a}{n+Ka}
\end{equation}
is covered, including $w_i(h)=(\operatorname{rank}_i(h)-k_r)_+$. This is the additive rank-penalty form of soft TACP \citep{liu2026tail}, with an explicit sufficient condition for using final calibration priors and independently fixed nuisance parameters. It is not a claim about every implementation or its conditional-coverage guarantees.

No movement of other-label scores is needed. Row-specific penalties can change same-class order, so the original $C_A$ argument cannot simply be transplanted. Even a map monotone in both arguments can break $C_A\subseteq C_E$ by merging ties, while remaining valid through $C_{\rm loo}$ (Appendix~\ref{app:general}). Also, normalized probabilities depend on all counts and are not generally expressible as $\phi(s_i(h),c_h)$. The transfer criterion, rather than that scalar formula, is what covers both families.

\section{Efficiency has no universal ordering}\label{sec:efficiency}
Containment $C_A\subseteq C_E$ for a decreasing function compares \emph{that function} with its own augmented reference. It does not compare either set with the set $C_0$ from uniform weights. The following result makes the distinction exact.

\begin{theorem}[The same valid count rule can help or hurt]\label{thm:efficiency}
Take $K=2$, $n=9$, $\alpha=1/10$, and the common fixed function
\[
 f(c)=\frac{1}{\max(c+1,5)}.
\]
There exist iid laws and positive feature-only bases for which ordinary count reuse respectively satisfies $C_f\subseteq C_0$ or $C_0\subseteq C_f$ pointwise, with strict expected-size inequalities. In both constructions, $\ind\{Y_*\in C_f\}=\ind\{Y_*\in C_0\}$ on every sample. The rule is nonconstant and nonincreasing, so it has the universal marginal guarantee in Theorem~\ref{thm:sharpcommon}.
\end{theorem}
\begin{proof}
Let $Y_i$ be iid with $\Pr(Y_i=1)=p\in(0,1)$. Let $X_i$ equal the positive row specified below for its class and use the same identity base $A_{ij}=(X_i)_j$ in both laws. Set $C=\sum_{i=1}^9\ind\{Y_i=1\}\sim\operatorname{Bin}(9,p)$. The feature maps below read $X$, not a query label. The rank is $k=9$.

For the helpful law, use base rows $(19,1)$ on class 1 and $(3,2)$ on class 2. Direct substitution in (\ref{eq:generalweight}) gives $C_0=\{1\}$ for a class 1 target, and $C_0=\{1,2\}$ for a class 2 target unless $C=9$, when it is empty. The count rule changes only class 2 targets at $C\in\{7,8\}$, removing the false label 1. Indeed $f(C)/f(9-C)=\max(10-C,5)/\max(C+1,5)$ crosses below $2/3$ exactly at $C=7$; at $C=9$ both candidates remain excluded. Hence both procedures have coverage $1-(1-p)p^9$, while
\[
 \mathbb E|C_0|-\mathbb E|C_f|=(1-p)\Pr\{C\in\{7,8\}\}>0.
\]

For the harmful law, replace the two base rows by $(3,2)$ and $(1,19)$. Uniform sets are the correct singleton, except that a class 1 target gets the empty set at $C=0$. The count rule changes only class 1 targets at $C\geq7$, adding the false label 2. The same weight-ratio boundary gives the comparison; all true-label decisions remain unchanged. Both cover with probability $1-p(1-p)^9$, and
\[
 \mathbb E|C_f|-\mathbb E|C_0|=p\Pr\{C\geq7\}>0.
\]
All comparisons use inclusive ties. The supplied ten-count decision tables verify each substitution exactly, rather than using simulation to estimate either probability.
\end{proof}

\begin{table}[ht]
\centering
\begin{tabular}{lrrr}
\toprule
Law, $p=4/5$ & Coverage, both (\%) & $\mathbb E|C_0|$ & $\mathbb E|C_f|$\\
\midrule
Helpful & 97.31564544 & 1.1463129088 & 1.0255169536\\
Harmful & 99.99995904 & 0.9999995904 & 1.5905575936\\
\bottomrule
\end{tabular}
\caption{Exact iid efficiency witnesses at nominal $90\%$. Coverage is identical \emph{within} each row, not equal between the two laws. Differences come exclusively from false labels. Values are exact rational probabilities, not empirical estimates.}\label{tab:efficiencywitness}
\end{table}

\paragraph{What this settles.} The helpful construction refutes the claim that uniform weights minimize expected size throughout the valid family. The harmful construction refutes the opposite universal ranking for this same rule. These are existence results, not natural-data performance claims or optimality over all independently fitted class-specific constants. The two-class example does not require $K\alpha\geq1$: that restriction belongs to the converse, not the sufficient validity or efficiency statements. Moreover, the feature bases here need not be accurate classifiers; classification perfection belongs to the separate failure constructions below.

\section{Applying the validity criterion}
For general separable scores, $C_E\cup C_{\rm loo}$ preserves every ordinary inclusion and the reference guarantee without assuming count monotonicity. The new \texttt{general\_scores/rank\_reference.py} implements this reference and guard from ordinary scores and each calibration row's score at one fewer own-class count. A fixed nondecreasing rule needs only ordinary calibration. For the normalized family, the earlier selective guard $C_E\cup C_A$ remains valid and its original code and checks are preserved. These safeguards use distinct full conformal references.

The criterion does not authorize calibration-supervised feature training, outcome-tuned penalties, arbitrary iterative transport, a changing label space, or selective splitting without their own arguments. Finite-sample validity also says nothing about relative set size; Theorem~\ref{thm:efficiency} and the matched fixed-weight controls address that separate question.

\section{How severe can invalid reuse be?}\label{sec:failure}
\paragraph{Transport as an application.} For positive $G_{ih}=\exp\{z_h(X_i)/\tau\}$, one cycle rescales each class column to weights $d_h=c_h+a$, then normalizes rows. It has (\ref{eq:generalweight}) with $A_{ih}=G_{ih}/\sum_iG_{ih}$ and $f(c)=c+a$. Thus $C_E\subseteq C_A$ for one cycle. Finite sample nesting need not survive later cycles; Appendix~\ref{app:threecycle} retains an exact counterexample. Coverage validity of full augmentation instead follows from its labeled symmetry, including at the unique positive-kernel limit.

\begin{proposition}[An exact iid fixed catalog failure]\label{thm:collapse}\label{eq:collapse}
Draw labels iid from probabilities $\pi$ on fixed $K\geq2$ classes. Set $X=e_Y$ and use the fixed feature-only logits $z_j(X)=(\log r)X_j$. With a unit pseudocount, one transport cycle, and $r\geq n+2$,
\begin{equation}\label{eq:collapse_new}
\Pr\{Y_*\in C_E\}=\sum_h\pi_h\Pr\{\operatorname{Bin}(n,\pi_h)>n-k_\alpha\}.
\end{equation}
At $K=20$, uniform labels, $n=200$, $\alpha=.1$, and $r=202$, original accuracy is one and confidence is $202/221=91.4027\%$, but prediction coverage is $0.266458\%$.
\end{proposition}
Appendix~C proves the formula by counting the other-class calibration rows strictly above the target; no Monte Carlo approximation enters it.

This is one specified mathematical population and predictor, not a claim of natural prevalence. Small continuous feature perturbations preserve it as an upper bound (Appendix~\ref{app:continuousfixed}). A stronger converged construction keeps the feature noise fixed rather than shrinking it: for uniform $K>1/\alpha$, fixed $a,\eta>0$, a fixed continuous feature law admits feature-only scores with correct-class boost
\begin{equation}\label{eq:logboost}
\beta_n\geq\log(n+1)+8\eta+\log\{32K(K-1)\}
\end{equation}
whose confidence approaches one while converged empirical coverage approaches zero. The sufficient rate is logarithmic, not claimed necessary. Appendix~\ref{app:logproof} gives the complete quota comparison and constants. The score sharpens with $n$; the statement is not inconsistency for a fixed classifier.

\paragraph{Why convergence does not cure the information asymmetry.} Calibration-fitted quotas leave the omitted target class short while allocating surplus mass to the other classes. At sufficiently sharp scores, non-target rows fit their own classes better than target-class rows, changing the rank even at the exact optimizer. Separately, when score contrast stays fixed but $K$ grows with $n$, a calibration observation sees the size-biased count $1+\operatorname{Pois}(\lambda)$ while an independent target sees $\operatorname{Pois}(\lambda)$. Appendix~\ref{app:limits} proves the resulting coverage limits uniformly over all positive iteration counts and convergence. These two regimes explain different failures; neither assumes that optimizing longer restores labeled symmetry.

\section{Evidence on actual model scores}\label{sec:evidence}
The original frozen study, later analyses of the same image pools, and exact constructions remain distinct. The natural experiments use the same calibration and target sampling law; their intervals condition on the extracted pools, not new environments.

\subsection{Frozen study: a modest but resolved finite cycle effect}
After development on EuroSAT \citep{helber2019} and Oxford Pets \citep{parkhi2012}, we froze model revisions, prompts, positive pseudocount $a=1$, sampling seeds, budgets, temperatures, and 13 comparison methods. We then extracted CLIP B/32, CLIP L/14, and SigLIP B/16 scores \citep{radford2021,zhai2023} on CIFAR10, Food101, and FGVC Aircraft \citep{krizhevsky2009,bossard2014,maji2013}. Hash selection without outcomes retained 6,000, 6,000, and 3,333 official test images. The same 256 temperature-fitting identities per dataset are excluded from every evaluation draw. Native and independently fitted temperatures are separate arms.

For each $n/K\in\{2,6,20\}$, 2,048 independent bags sample $n+1$ observations with replacement from the fixed remaining pool; all possible target roles are integrated for one cycle. The first 512 bags also evaluate three cycles. A whole bag is the independent unit, not an individual role. Empirical Bernstein intervals \citep{maurer2009} use a fixed $.05$ allocation over 648 endpoints. Duplicate image identities are retained as required by this conditional iid sampling law.

\begin{table}[t]
\centering\footnotesize\setlength{\tabcolsep}{3pt}
\caption{All $27$ primary native temperature, one cycle conditions at nominal $90\%$. Values are mean coverage percentages with simultaneous $95\%$ intervals from the fixed family of $648$ endpoints. Each cell integrates all target roles in $2,048$ independent bags. Interval endpoints are rounded outward. Eight upper bounds lie below $90\%$; all eight occur at $n=2K$.}
\label{tab:confirmation}
\begin{tabular*}{\linewidth}{@{\extracolsep{\fill}}llccc@{}}\toprule
Encoder & Data & $n/K=2$ & $n/K=6$ & $n/K=20$\\\midrule
CLIP B/32 & CIFAR10 & 88.38 [86.78, 89.99] & 89.30 [87.93, 90.66] & 89.78 [88.50, 91.06]\\
CLIP B/32 & Food101 & 88.53 [87.20, 89.87] & 89.39 [88.11, 90.67] & 89.78 [88.52, 91.03]\\
CLIP B/32 & Aircraft & 87.10 [85.72, 88.48] & 88.87 [87.58, 90.16] & 89.69 [88.43, 90.95]\\
CLIP L/14 & CIFAR10 & 88.34 [86.75, 89.93] & 89.32 [87.96, 90.68] & 89.82 [88.54, 91.10]\\
CLIP L/14 & Food101 & 88.22 [86.87, 89.58] & 89.33 [88.05, 90.61] & 89.77 [88.51, 91.02]\\
CLIP L/14 & Aircraft & 86.95 [85.56, 88.34] & 88.76 [87.47, 90.06] & 89.65 [88.39, 90.91]\\
SigLIP B/16 & CIFAR10 & 89.06 [87.53, 90.59] & 89.53 [88.18, 90.88] & 89.84 [88.57, 91.12]\\
SigLIP B/16 & Food101 & 88.51 [87.17, 89.86] & 89.44 [88.16, 90.71] & 89.79 [88.53, 91.04]\\
SigLIP B/16 & Aircraft & 87.28 [85.92, 88.65] & 88.88 [87.59, 90.17] & 89.67 [88.41, 90.93]\\
\bottomrule\end{tabular*}
\end{table}

At nominal $90\%$, eight of 27 primary native one cycle conditions have simultaneous upper coverage limits below $90\%$, all at $n=2K$ (Table~\ref{tab:confirmation}). Across those nine small-budget cells, coverage is $86.95\%$--$89.06\%$: approximately $10.94$--$13.05$ misses per 100 rather than the nominal ten. These are modest deficits. At $n/K=6$ or 20, this simultaneous family resolves none; absence of detection is not equality to the nominal level. All conditions, including the unresolved SigLIP/CIFAR10 cell, remain in the table.

The theorem's comparison explains the direction without demanding that every cell violate nominal coverage. Rank rounding and ties can provide surplus coverage. In each one cycle bag, full augmentation's acceptance fraction minus the ordinary fraction is exactly the fraction of augmented-accepted target roles lost through count reuse. Those roles are never gained in the opposite direction. Appendix~\ref{app:erosion} gives the identity and efficient audit; it is an offline labeled calculation, not a guarantee inferred from unknown target labels.

\subsection{Later convergence analysis: a larger natural score failure}
The convergence follow-up was designed after those outcomes were known. It reuses all nine native score caches and all three budgets, with 1,024 independently sampled calibration-plus-one-target episodes per cell. Cycles 1, 3, 10 and the certified limit share these identities and a separate 432-endpoint confidence allocation. It is not another untouched dataset study.

\begin{figure}[ht]
\centering\includegraphics[width=\linewidth]{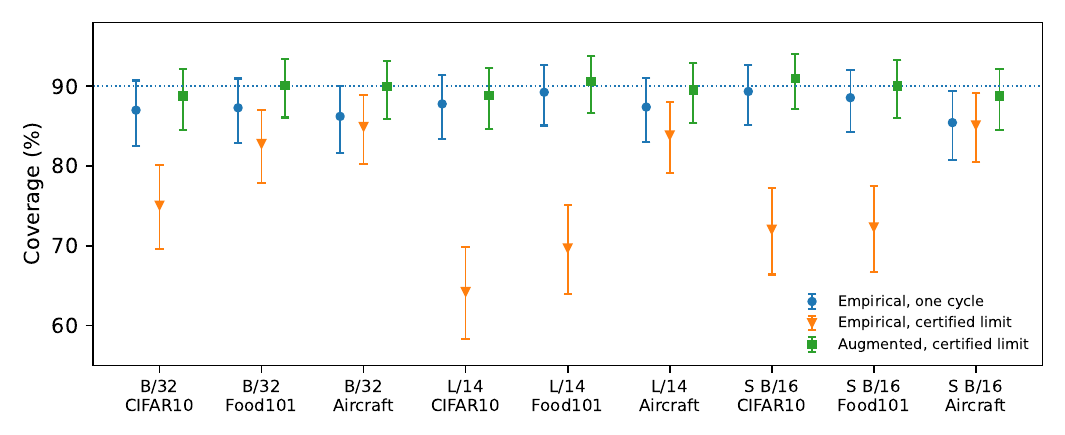}
\caption{\textbf{Coverage at native model temperatures.} Later analysis on the known pools at $n=2K$, with a fixed native temperature and no artificial score sharpening. Each point uses 1,024 independent target episodes; bars use the separate simultaneous allocation. This figure does not replace the frozen one cycle study. All nine pairs and both empirical and augmented limits are shown.}
\label{fig:naturalaudit}
\end{figure}

At the certified limit, 14 of 27 empirical conditions resolve undercoverage at nominal $90\%$: all nine at $n=2K$, five at $n=6K$, and none at $n=20K$. Thus the convergence finding is not confined to two calibration observations per class. At $n=2K$, empirical coverage ranges from $64.26\%$ to $85.16\%$, with a descriptive median of $75.10\%$ across the nine fixed cells. Those means correspond to $1.48$--$3.57$ times the nominal miss budget. These are fixed-cell summaries, not uncertainty estimates over new environments.

For CLIP L/14 on CIFAR10, original classification accuracy is $95.33\%$ on the evaluation pool. At $n=20$, converged empirical coverage is $64.26\%$ with simultaneous interval $[58.31\%,69.92\%]$; the augmented estimate is $88.87\%$ with interval $[84.64\%,92.31\%]$. These are native pretrained scores, distinct from the constructed $0.266458\%$ example. The augmented procedure has no significant deficit in this family; that does not prove exact nominal coverage. The complete convergence tables are in Appendix~\ref{app:empirical}. Table~\ref{tab:confidencecheck} also reports an above nominal point estimate in the highest confidence subgroup, with a wide interval; aggregate undercoverage is not undercoverage in every subgroup.

\paragraph{Numerical certification.} All 55,296 returned limit fits have valid score enclosures and resolved target ranks under the stated arithmetic contract. Before counting ranks, the certificate bounds error relative to the exact positive rounded-kernel limit. Any unresolved decision would widen the statistical interval. Appendix~\ref{app:numerics} gives the contraction argument, rounding assumptions, and alternative-formulation checks.

\subsection{Valid alternatives remain strong}
A separate frozen schedule computes complete sets for 128 calibration-plus-16-query episodes per condition. All 13 methods share total calibration/fitting label budgets, with temperature-fitting labels charged separately. At native temperature and nominal $90\%$, empirical three cycle transport averages 4.767 labels and $87.70\%$ coverage; full augmentation averages 5.237 and $90.12\%$; uniform transport averages 5.254 and $90.04\%$. These descriptive means are not a significance test of efficiency. Every method and condition remains in the artifact, with all 13 native averages in Appendix~\ref{app:empirical}.

A later directional study evaluates raw and column-normalized bases and every exponent in $\{-1,-.5,0,.5,1\}$ on the same pools. Negative exponents obey the proved containment but do not beat constant weights on mean set size. Those original aggregate results do not establish universal dominance. The expanded comparison below separates the potential value of class adjustment from the incremental effect of live count reuse.

\subsection{Expanded safe-rule search with independent selection}\label{sec:expanded}
This new known-pool analysis fixes 55 candidates: uniform, eight powers, ten clipped functions, twelve step functions, nine shrinkage functions, and fifteen class-specific functions. The last group permits different fixed powers or thresholds per class and confidence-based multiplicative constants. All functions are positive and nonincreasing conditional on the 256 previously excluded fitting rows. The final evaluation labels never choose the function. The guard therefore equals the ordinary set; no candidate refit is needed for these rules.

For each of the nine existing caches, three budgets, two feature bases, and two nominal levels, 32 fitting-pool bootstrap episodes choose a global winner and a winner within each family by mean complete-set size. A separate 512-episode evaluation computes all candidate sets with 16 queries per episode. Every fitting label is charged, also to the baselines. Each candidate is paired with a static control obtained by evaluating that same function at the independently estimated counts $n\widehat\pi_h$ and then freezing its weights. This control isolates live counts from class information learned independently. Appendix~\ref{app:utility} states the candidate formulas, selection rule, sampling scope, and full per-cell files.

\begin{table}[ht]
\centering
\begin{tabular}{llrrrr}
\toprule
Base & Nominal (\%) & Uniform & Selected & Frozen match & Lower/equal/higher\\
\midrule
Raw &90&7.0563&6.6677&6.5475&12/4/11\\
Column &90&5.7211&5.7728&5.7141&2/19/6\\
Raw &95&10.4201&9.8754&9.6776&15/0/12\\
Column &95&8.1824&8.2324&8.1522&4/16/7\\
\bottomrule
\end{tabular}
\caption{Expanded search: equal-cell mean set sizes across 27 cache/budget cells. Each selected rule was chosen using only the separate fitting pool. The last column counts its evaluation size relative to uniform, not significant wins. The matching frozen rule is selected by the same independent choice, not by evaluation results.}\label{tab:expanded}
\end{table}

At nominal $90\%$, raw-base selection gives smaller point estimates than uniform in 12 of 27 cells, but larger ones in 11; its average coverage is $90.40\%$. Column-base selection improves two cells and worsens six, with $90.31\%$ average coverage. No selected comparison resolves a size advantage in the conservative simultaneous empirical Bernstein family. In all four aggregate rows, matching independently frozen weights are smaller than the live-count selection. Thus the exact theorem establishes \emph{possibility}; these data neither prove uniform dominance nor establish a statistically resolved deployment advantage from live counts. All family winners, all 55 candidates, and all per-cell differences are reported.

\subsection{Conf-OT proper: completed native benchmark subset}\label{sec:native}
We separately executed the released extractor and benchmark driver at commit \texttt{1a29ced}, using its DTD and FGVC Aircraft test splits \citep{cimpoi2014,maji2013,silva2025}. CLIP ViT-B/16, original prompts and preprocessing, stratified half calibration, three transport cycles, no pseudocount, and the released LAC/APS/RAPS settings were retained. Only the documented dataset-root configuration changed. All 12 commands completed, giving 24 dataset/method/score/level cells and 480 split records. Both possible result interpretations were specified before execution.
\begin{table}[t]
\centering\small
\caption{Completed native Conf-OT subset. Median coverage (\%) across 20 released stratified splits per cell. Unadapted is the repository's \texttt{none} branch. These are descriptive medians, not simultaneous confidence bounds. Full sizes and split ranges are in Appendix~\ref{app:native}.}\label{tab:native-main}
\begin{tabular}{llrrrr}
\toprule
 & & \multicolumn{2}{c}{Nominal 90\%} & \multicolumn{2}{c}{Nominal 95\%}\\
Data & Score & Unadapted & Conf-OT & Unadapted & Conf-OT\\
\midrule
DTD & LAC & 89.95 & 89.48 & 95.27 & 94.33\\
DTD & APS & 90.25 & 89.89 & 95.45 & 94.62\\
DTD & RAPS & 90.25 & 89.95 & 95.27 & 94.62\\
Aircraft & LAC & 89.76 & 89.97 & 95.12 & 94.82\\
Aircraft & APS & 89.91 & 89.85 & 94.97 & 95.29\\
Aircraft & RAPS & 89.94 & 89.88 & 95.00 & 95.12\\
\bottomrule
\end{tabular}
\end{table}

Conf-OT median coverage is 89.48--89.97\% at nominal 90\% and 94.33--95.29\% at nominal 95\% (Table~\ref{tab:native-main}). Median set sizes are 17.61--29.54\% smaller than the released unadapted baseline's across the twelve paired cells. These descriptive comparisons neither establish exact nominal coverage nor a size advantage at matched true coverage; shared images across repetitions are not independent binomial trials.

The split fixes each class count before choosing within-class identities: DTD always has 18 calibration images per class, so its fitted prior is exactly uniform. Aircraft has 16 or 17, also fixed across splits. Thus this benchmark does not vary the prior through the iid self-count mechanism studied above. Fixed quotas do not by themselves prove every pooled-rank guarantee. The subset changes several design components relative to our iid study, so it is not a causal isolation of stratification. Appendix~\ref{app:native} gives the precise differences, all outcomes, and provenance. The reproduction covers this two-dataset subset of the 15-dataset Conf-OT benchmark; native TIM, TransCLIP, FCA, and SCA-T are outside its scope.

\section{Related work, limitations, and conclusion}\label{sec:conclusion}
The full conformal reference is classical \citep{shafer2008}; label conditional ranks \citep{vovk2012,ding2023} and stability-based outer sets \citep{ndiaye2022} are established alternatives. Valid adaptive transductive inference \citep{gazin2024} does not assert that every label-dependent transformation is safe. Work on conformal model selection \citep{liang2026} and tuning bias \citep{zeng2025} concerns different reuse mechanisms. Building on class frequency adjustment \citep{menon2021}, we characterize its ordinary-rank validity boundary when \emph{final calibration labels} determine the correction.

\paragraph{Class frequency methods and scope.} \citet{liu2026tail} use TACP and soft TACP to improve head/tail or class-level coverage balance. Our new corollary covers its nonnegative additive rank-penalty form when the prior uses final calibration counts and all other fitted choices meet the stated information restrictions. It does not reproduce their experiments or certify adaptive tuning. \citet{xie2025} study unseen labels and frequency-dependent selective splitting with weighted ranks; an open catalog and a changed sampling design remain outside this paper. The new result is a count-transfer sufficient criterion with sharp common-map specializations, not a characterization of every calibration-trained scoring rule. The full conformal leave-self-out reference is classical; the score comparison is the argument that justifies ordinary reuse.

The natural data findings are conditional on finite image pools; confidence intervals cover calibration randomness, not unseen environments. Later analyses remain retrospective. Section~\ref{sec:native} separately reports the released Conf-OT code on its two default benchmark datasets under the original stratified design. Its near-nominal medians do not contradict the iid counterexamples, and neither study certifies arbitrary deployments. Original history, competitive baselines, and the small historical reconstruction discrepancy remain in the evidence record.

Validity and efficiency require separate answers. The transfer criterion covers additive and normalized scores, with sharp common-map boundaries in the stated regime. Even one fixed safe rule can improve or worsen expected size relative to uniform weights at unchanged coverage. The expanded comparisons preserve strong independently fitted controls. The conclusion is not that count reuse is always worthwhile or never worthwhile: its validity can be characterized, while its utility must be established for the intended population and comparator.
\label{maintextend}

\clearpage
\section*{Reproducibility statement}
The artifact contains the original image studies, their registered conditions, subsequent analyses, source and input hashes, and exact mathematical checks. The expanded safe-rule study uses separate fitting labels for selection. The author executed the native Conf-OT notebook on DTD and Aircraft; the archive contains all 24 benchmark cells, all 480 split records, feature caches, commands, receipts, and job logs. Its launcher matches the distributed version. Internal computational audits reconstruct the native summaries from saved outputs and check input, output, protocol, and source-change records. The complete proof of the central theorem is in the main text. Historical manuscripts and study records are retained in the archive.
\section*{AI use statement}
The author led the research and evaluation, with generative AI assistance in mathematical development, coding, analysis, and writing. The author remains responsible for all claims.
\section*{Ethics statement}
This work uses public research image datasets and synthetic mathematical examples. It recruits no participants. The evidence artifact does not redistribute raw photographs or pretrained weights; reacquisition remains subject to original terms. The earlier custom image study removes the Aircraft attribution strip; the native reproduction instead retains the released image preprocessing. Both input protocols are documented separately. The analysis concerns specified method configurations and sampling assumptions. Marginal coverage is not an individual or clinical safety guarantee.

\clearpage
\appendix
\section{How to read the supporting material}\label{app:map}
The central validity theorem and its complete proof are in Section~\ref{sec:boundary}. Appendix~\ref{app:guard} details the information restrictions and implementation. Appendices~\ref{app:continuousfixed}--\ref{app:numerics} prove the secondary iid, asymptotic, and numerical statements. Appendix~\ref{app:empirical} records sampling and full comparisons, and Appendix~\ref{app:verification} documents the computational checks. The artifact also preserves the original 60 page development record, historical experiments, and additional constructions.

\begin{table}[h]\centering\small
\begin{tabularx}{\linewidth}{@{}p{.27\linewidth}p{.33\linewidth}X@{}}\toprule
Claim & Quantification & Evidence type\\\midrule
Count-transfer sufficiency & Every $n$ and admissible exchangeable law; fixed feature marks and maps satisfying \eqref{eq:transfer} & Leave-self-out reference containment, main text.\\
Common-map converses & Fixed $K,\alpha$, map; $K\alpha\geq1$; separable or normalized specialization & Exact balanced construction, main text; normalized continuous detail below.\\
Iid fixed-catalog collapse & Specified iid law and predictor at a fixed $n$; sequence sharpens for the vanishing limit & Binomial calculation, main text and Appendix~\ref{app:continuousfixed}.\\
Noisy converged collapse & Fixed feature law, $K$, $a$, noise; score boost grows as stated & Deterministic quota lemma and a count tail, Appendix~\ref{app:logproof}.\\
Growing catalog limits & $K\to\infty$, $n/K\to\lambda$; fixed contrast and bounded intensities & Concentration, contraction, and dependent quantile argument, Appendix~\ref{app:limits}.\\
Native model findings & Conditional on each fixed image pool, under declared repeated sampling & Complete frozen or explicitly later confidence families, Appendix~\ref{app:empirical}.\\\bottomrule
\end{tabularx}
\caption{Scope and supporting evidence for each result. Pointwise identities, distributional theorems, numerical enclosures, and confidence intervals have distinct roles.}
\end{table}

\paragraph{Continuous normalized converse.}
For the balanced witness in Theorem~\ref{thm:sharpcommon}, let $q=f(m-1)/f(m)<1$ and $R>1$. The ratio of an other-class true probability to the target true probability is
$\Lambda=(Rq+K-1)/[q(R+K-2+q)]>1$.
Perturb each log base entry by independent uniform noise in $[-\eta,\eta]$, with $8\eta<\log\Lambda$ and $2\eta<\log R$. Normalization changes each probability by a factor in $[e^{-2\eta},e^{2\eta}]$, so all sufficient cross-class comparisons remain strict. The original classifier remains correct, and its confidence is at least $R/[R+(K-1)e^{2\eta}]$. Taking $R$ large and then $\eta$ small yields any prescribed confidence below one. The label law remains balanced and exchangeable, not iid; the conditional feature law has a density. This strengthens only the normalized converse, not every arbitrary scalar-map witness.

\section{Information restrictions, selective augmentation, and scope tests}\label{app:guard}
The base is a measurable function of pooled features and separately conditioned fitting information. The feature map, class catalog, count functions, and all hyperparameters must be fixed without the final calibration labels. Permuting calibration and target rows must permute the output rows. The conditional exchangeability requirement does not follow automatically for an adaptively chosen query pool. A base may be highly predictive of labels; being feature only does not require statistical independence from labels.

The direction proof remains valid for any fixed class specific positive functions. Nonmonotone functions can have different directions at different realized counts. The common-function converse fixes the same $f$ across all sample sizes. Its $m=1$ witness has an unobserved target class in calibration, with all weights still positive; at $m\geq2$ every class is observed. Constant rescaling of all weights cancels, but tuning a function from final calibration performance is not such a harmless rescaling.

\paragraph{The normalized all-count reference.}
The general transfer theorem uses $C_{\rm loo}$; the earlier normalized implementation instead uses $C_A$. Its original containment has a separate elementary proof. For a candidate $h$, put $\rho_h=f_h(c_h+1)/f_h(c_h)$. Then
\begin{equation}\label{eq:generalupdate}
p_i^{(h)}(j)=\frac{p_i(j)\rho_h^{\ind\{j=h\}}}{1+(\rho_h-1)p_i(h)}.
\end{equation}
For the query and same-label calibration rows the map $v\mapsto\rho_h v/[1+(\rho_h-1)v]$ is strictly increasing. If $\rho_h\leq1$, their probabilities decrease while other-label true probabilities increase. The number of strictly better calibration probabilities cannot decrease, so $h\in C_A$ implies $h\in C_E$. The implications reverse for $\rho_h\geq1$. This extra structure is not assumed for arbitrary additive or separable scores.

\paragraph{Selective guard.} Put $\rho_h=f_h(c_h+1)/f_h(c_h)$ and compute the ordinary set first. If $\rho_h\leq1$, the normalized update above gives $\ind\{h\in C_A\}\leq\ind\{h\in C_E\}$. If $\rho_h>1$, the reverse inclusion holds. Therefore the rule
\[
h\in C_G\quad\Longleftrightarrow\quad h\in C_E\ \text{or}\ [\rho_h>1\text{ and }h\in C_A]
\]
is exactly $C_E\cup C_A$. It preserves the full reference coverage while retaining every ordinary inclusion. For a query batch, one hypothetical map for $h$ can be shared across all excluded query candidates. No fit is needed for a nonincreasing function, nor when all ordinary candidates of that class are already included. This is not a minimum-size procedure and may be more conservative than $C_A$.

The implementation \texttt{reviewer\_audit/count\_weight\_guard.py} accepts the positive base, final calibration labels, and fixed function values $f_h(0),\ldots,f_h(n+1)$. It offers ordinary, full augmented, and guarded modes. Under a standard binary64 positive-sum error model, probability intervals certify unambiguous comparisons; ambiguous comparisons fall back to exact rational arithmetic on the supplied binary64 inputs. Duplicate rows with the same candidate label retain exact ties. Invalid dimensions, nonpositive values, nonfinite arithmetic, or subnormal intermediates are rejected rather than silently altered. The information provenance and fixed-function requirements must be checked when configuring the inputs.

For a calibration supervised scoring pipeline outside this family, a valid full reference must refit \emph{every label-dependent stage} under each candidate, or use a separately justified method. The general sufficient principle is rank dominance: if the implemented count of strictly better calibration scores is no greater than that of an exchangeable reference, its set contains the reference. The count theorem derives that dominance from a local weight update; it does not assert it for arbitrary training. Reference unions provide a standard validity safeguard.

\paragraph{Worked example and negative scope fixtures.} For Figure~\ref{fig:mechanism}, the counts are $(5,4)$, weights $(6,5)$, and candidate-augmented weights $(6,6)$. Independent Fraction arithmetic checks every table entry, the ranks $9$ and $4$, and both candidate decisions. Additional fixtures show that a step beyond all reachable counts cannot invalidate a fixed-$n$ rule, that the balanced witness need not exclude when $K\alpha<1$, and that a base deliberately memorizing calibration labels can fail even with constant weights. These fixtures check the stated quantifiers and information restrictions.

\section{Continuous perturbations of the finite iid counterexample}\label{app:continuousfixed}
We first give the exact proof of Proposition~\ref{thm:collapse}.
\begin{proof}
Fix target $h$, set $P=n+1$ and $N_j=c_j+\ind\{j=h\}$. Up to a common factor, column multipliers are $b_h=N_h/[P+(r-1)N_h]$ and $b_j=(N_j+1)/[P+(r-1)N_j]$ for $j\ne h$. The numerator of $b_j-b_h$ is $P(N_j+1)+(r-1-P)N_h>0$. For a row of class $j$, its true probability $rb_j/[\sum_\ell b_\ell+(r-1)b_j]$ increases with $b_j$. Exactly $n-c_h$ calibration scores beat the target; same-class scores tie. Equation~(\ref{eq:rank}) gives acceptance iff $c_h>n-k_\alpha$, and iid counts give the formula.
\end{proof}
At $r=P+1$, all non-target multipliers equal $1/P$, while the target multiplier is $\kappa/P$, with $\kappa=N_h/(N_h+1)$. An other-class true probability divided by the target probability exceeds one by
\[
\frac{K-1+\kappa}{N_h(P+K-1+\kappa)}\ \geq\ \frac{K-1}{P(P+K)}=:\Delta_{n,K}>0.
\]
Multiplying each kernel entry by a factor in $[e^{-\eta},e^\eta]$ changes the one-cycle probabilities by at most $[e^{-4\eta},e^{4\eta}]$: one factor comes from the kernel, one from column normalization, and the remaining normalization gives another two. Thus an across-row probability ratio loses at most $e^{-8\eta}$. If
\[
8\eta<\log(1+\Delta_{n,K}),
\]
every other-class calibration score remains strictly above the target. Same-class ties can split, causing more exclusions but never invalidating the binomial \emph{upper} bound.

For instance, take $X=(\log202)e_Y+\xi$, with iid continuous $\xi_j\sim\operatorname{Unif}[-10^{-5},10^{-5}]$, $K=20$, $n=200$, and the identity logit map. Its original classifier is perfect, native confidence exceeds $91.4\%$, and empirical one-cycle coverage is at most $0.266458\%$. The exact equality example and this continuous upper bound must not be interchanged.

For uniform labels, the exact probability in Proposition~\ref{thm:collapse} is the rational number
\[
K^{-n}\sum_{j=n-k_\alpha+1}^{n}\binom nj(K-1)^{n-j}.
\]
The numerator and denominator are evaluated with integers. For any fixed $\pi_{\max}<\alpha$, choose $0<\epsilon<\alpha-\pi_{\max}$. At all sufficiently large $n$, acceptance implies $c_{Y_*}/n>\alpha-\epsilon$, and Hoeffding's inequality bounds coverage by $\exp[-2n(\alpha-\epsilon-\pi_{\max})^2]$. The accompanying predictor sharpens with $n$; this is not a fixed-classifier inconsistency statement.

\section{Full proof of the noisy converged collapse}\label{app:logproof}\label{app:sharpquota}
We first prove a deterministic quota lemma, then give the sampling argument. This proof replaces reliance on an optimization tolerance with a statement about the exact positive-kernel solution.

Fix $P$ rows, $K\geq2$, full label counts $N_j$, and target class $h$ with $N_h>0$. Positive quotas $Q_j$ sum to $P$ and satisfy $Q_j-N_j\geq s>0$ for $j\ne h$, and $\delta=N_h-Q_h>0$. Let $G_{ij}=\exp\{\beta\ind\{Y_i=j\}+\xi_{ij}\}$, with $|\xi_{ij}|\leq\eta$, and let $p$ be its exact row-stochastic fit with column sums $Q$.
\begin{lemma}[Exponential non-target error bound]\label{lem:sharpquota}
Put $u=e^{-\beta+2\eta}$ and $D=s-Pu$. If $D>0$, then every row of a non-target class has error at most
\begin{equation}
U=\frac{(K-2)e^{2\eta}u}{D}\left(\delta+\frac{P^2u^2}{D}\right)
       +\frac{Pu^2}{D}.
\label{eq:sharpU}
\end{equation}
Every target-class row has error at least $e^{-4\eta}\delta/N_h$. Consequently, $U<e^{-4\eta}\delta/N_h$ forces all non-target true probabilities strictly above all target-class true probabilities.
\end{lemma}
\begin{proof}
Write $p_{ij}=G_{ij}b_j/\sum_kG_{ik}b_k$ with $b_j>0$. Select a non-target index $j_0$ minimizing $b_j$ among $j\ne h$. A row of non-target class $k\ne j_0$ satisfies
\[
p_{i,j_0}/p_{i,k}\leq u\,b_{j_0}/b_k\leq u.
\]
Its probability sent to $j_0$ is therefore at most $u$. Non-target rows with label $j_0$ contribute at most $N_{j_0}$ in total. Thus target-class rows must send
\[
A_{j_0}:=\sum_{i:Y_i=h}p_{i,j_0}\geq Q_{j_0}-N_{j_0}-Pu\geq D.
\]
On a target row, $p_{ih}/p_{i,j_0}\geq u^{-1}b_h/b_{j_0}$. Summing and using $A_h\leq P$ yields $b_h/b_{j_0}\leq Pu/D$. It follows that the total non-target mass entering column $h$ is bounded by
\[
I_h:=\sum_{i:Y_i\ne h}p_{ih}\leq Pu\,b_h/b_{j_0}\leq P^2u^2/D.
\]
Column conservation now gives $\sum_{j\ne h}A_j=\delta+I_h$. For two non-target columns on target-class rows, $p_{ij}/p_{i,j_0}\geq e^{-2\eta}b_j/b_{j_0}$. Hence
\[
\frac{b_j}{b_{j_0}}\leq e^{2\eta}\frac{A_j}{A_{j_0}}
\leq e^{2\eta}\frac{\delta+I_h}{D}.
\]
A row of non-target class $k$ sends at most $u b_j/b_k$ to any wrong class $j$. Sum the $K-2$ non-target wrong columns and the target column, using $b_k\geq b_{j_0}$ and the preceding bounds, to obtain \eqref{eq:sharpU}.

For target-class rows, the error odds are
\[
\frac{1-p_{ih}}{p_{ih}}=\sum_{j\ne h}e^{-\beta+\xi_{ij}-\xi_{ih}}b_j/b_h.
\]
Any two such odds, and therefore any two such errors, differ by at most a factor $e^{4\eta}$. Their total error is at least $N_h-Q_h=\delta$. Each is consequently at least $e^{-4\eta}\delta/N_h$. Comparing the two bounds proves separation.
\end{proof}
The lemma allows arbitrary full counts and bounded within-class variation. No iid assumption, entropy approximation, or numerical stopping rule enters it. Its hypotheses require a positive surplus in \emph{every} non-target quota; it does not cover arbitrary quota configurations.

\subsection{An explicit logarithmic sufficient boost}
\label{app:logboost}
Set $s_0=1/(2K)$ and $\delta_0=1-3/(2K)>0$. Suppose the quota assumptions hold with $s\geq s_0$ and $\delta\geq\delta_0$. We show that
\begin{equation}
\beta\geq\log P+8\eta+\log\{32K(K-1)\}
\label{eq:sharpboostappendix}
\end{equation}
is sufficient for strict separation, including at $K=2$. Write $t=Pu$. Then
\[
t\leq\frac{e^{-6\eta}}{32K(K-1)}=\frac{s_0e^{-6\eta}}{16(K-1)},\qquad D\geq s_0/2.
\]
The incoming mass bound is $I_h\leq 2t^2/s_0\leq s_0/2\leq\delta$. Multiplying \eqref{eq:sharpU} by $P$ gives
\[
PU\leq \frac{4(K-2)e^{2\eta}\delta t}{s_0}+\frac{2t^2}{s_0}
\leq\frac{\delta e^{-4\eta}}4+\frac{\delta e^{-4\eta}}4
<\delta e^{-4\eta}.
\]
For the second inequality, substitute the displayed bound on $t$ and use $\delta\geq\delta_0\geq s_0$; the first summand vanishes at $K=2$. Since $N_h\leq P$, Lemma~\ref{lem:sharpquota} applies. The constants are conservative and are not a lower bound on the required contrast.

\paragraph{Exact finite iid bound with a total pseudocount of one.}
Take the empirical prior with $a=1/K$, so the converged column quotas are, identically for every count vector,
\[
Q_j=P\frac{c_j+1/K}{n+1}=c_j+1/K.
\]
If the target class is $h$, then $N_j=c_j$ for $j\ne h$, $N_h=c_h+1$, and
\[
Q_j-N_j=1/K\ (j\ne h),\qquad N_h-Q_h=1-1/K.
\]
They satisfy the sufficient bounds above without a concentration event. If \eqref{eq:sharpboostappendix} holds, the target can be accepted only if $c_h>n-k_\alpha$: otherwise at least $k_\alpha$ other-label calibration probabilities are strictly higher. For any iid class law $\pi$, therefore,
\begin{equation}
\Pr\{Y_*\in C_E^{(\infty)}\}\leq\sum_h\pi_h\Pr\{\operatorname{Bin}(n,\pi_h)>n-k_\alpha\}.
\label{eq:sharpbintail}
\end{equation}
This is an upper bound rather than equality because within-class noisy scores may cause additional exclusions. With uniform $K=20$, $n=200$, and $\alpha=.1$, it is $0.266458\%$, for any fixed $\eta$ and the sufficient boost. This is a \emph{different} construction from the unit-pseudocount, one-cycle example with $91.4\%$ native confidence. The new bound holds at exact convergence with nonvanishing bounded noise and a specified total pseudocount of one.

For fixed $\pi_{\max}<\alpha$, the binomial upper bound decreases exponentially by the same Hoeffding argument as Appendix~\ref{app:continuousfixed}. Original correct-label confidence obeys
\begin{equation}
p_{\rm native}(Y\mid X)\geq[1+(K-1)e^{-\beta+2\eta}]^{-1}
\geq1-\frac{e^{-6\eta}}{32KP}.
\label{eq:sharpconfidence}
\end{equation}
Thus a logarithmic boost suffices while confidence approaches one at least as fast as $1-O(1/n)$.

\paragraph{Every fixed positive pseudocount.}
For a general fixed $a>0$ and uniform labels, reuse the fixed continuous law in the following explicit construction: $X=\mu e_Y+\eta V$ for fixed $\mu>2\eta$ and independent $V_j\sim\operatorname{Unif}[-1,1]$. The feature-only score
\[
z_{n,j}(X)=\beta_n\ind\{j=\arg\max X\}+X_j-\mu\ind\{j=\arg\max X\}
\]
has class boost $\beta_n$ and offsets $\eta V_j$. Neither the feature law, its noise, $K$, nor $a$ changes with $n$. For $K>1/\alpha$ set
\[
\epsilon=\min\left\{\frac{\alpha-1/K}{2},\frac{1}{2K\max(1,|1-Ka|)}\right\},\qquad
E_n=\{\max_j|c_j/n-1/K|\leq\epsilon\}.
\]
With probability at least $1-2Ke^{-2n\epsilon^2}$, for every class
\[
Q_j-c_j=\frac1K+\frac{n(1-Ka)}{n+Ka}(c_j/n-1/K)\in[1/(2K),3/(2K)].
\]
Consequently the non-target surplus and target deficit bounds $s_0,\delta_0$ hold simultaneously for every possible target label. For all $n$ satisfying $n(\alpha-1/K-\epsilon)\geq2-\alpha$, also $c_h\leq n-k_\alpha$ on $E_n$. If
\[
\beta_n\geq\log(n+1)+8\eta+\log\{32K(K-1)\},
\]
every target on $E_n$ is excluded and
\begin{equation}
\Pr\{Y_*\in C_E^{(\infty)}\}\leq2Ke^{-2n\epsilon^2}\longrightarrow0.
\label{eq:sharpuniformtail}
\end{equation}
This proves the general-pseudocount conclusion stated in Section~\ref{sec:failure}. It does not show that a logarithmic rate is necessary, that every fixed model has this limit, or that any artificial boost was applied to the natural datasets. Full augmentation remains valid by the reference rank argument in Section~\ref{sec:setup} for its permutation-equivariant ideal fit.

\section{Complementary growing catalog theory}\label{app:limits}
These results are supplementary to the main common-function characterization. They keep score contrast fixed while the catalog grows, unlike the fixed-catalog sharpening construction. Labels in the first two cases are iid uniform on the catalog; the last case uses the explicitly stated unequal probabilities.

Write $n/K\to\lambda\in(0,\infty)$, $c_h=\sum_{i\leq n}\ind\{Y_i=h\}$, $m_h=(c_h+a)/(n+Ka)$, and
\begin{equation}\label{eq:informative}
X_i=\mu e_{Y_i}+\eta U_i,\qquad U_{ih}\ \text{iid}\sim\operatorname{Unif}[-1,1],\qquad G_{ih}=\exp(X_{ih}).
\end{equation}
Here $a,\eta>0$ and finite $\mu$ are fixed; $\mu>2\eta$ makes the original classifier perfect. There is one target. Let $T_K\in\{1,2,\ldots\}\cup\{\infty\}$ be any measurable iteration choice. Put $F_\lambda(t)=\Pr\{\operatorname{Pois}(\lambda)\leq t\}$, with zero for $t<0$. When an integer $r\geq1$ satisfies
\begin{equation}\label{eq:phase}
F_\lambda(r-2)<\alpha<F_\lambda(r-1),
\end{equation}
the following limits hold.
\begin{proposition}[Occupancy coverage limits]\label{thm:flat}\label{thm:continuous}\label{thm:mixed}
For uniform labels and identical score rows, empirical coverage tends to $1-F_\lambda(r-1)$ for any positive number of cycles or convergence; full augmentation tends to $1-F_\lambda(r-2)$. For the continuous model (\ref{eq:informative}), assume
\begin{equation}\label{eq:bandgap}
e^{2\eta}<\min\{(r+a)/(r-1+a),(r+1+a)/(r+a)\}.
\end{equation}
Uniformly over iteration choices $T_K$, empirical coverage then tends to
\begin{equation}\label{eq:continuouslimit}
1-F_\lambda(r-1)-\frac{\lambda}{r}\{\alpha-F_\lambda(r-2)\}.
\end{equation}
For unequal class probabilities, let $\lambda_{K,h}=n\pi_{K,h}$ be uniformly bounded, with $K^{-1}\sum_h\delta_{\lambda_{K,h}}\Rightarrow H$ and mean $\lambda$. Draw $L\sim H^*$, where $H^*(d\ell)=\ell H(d\ell)/\lambda$, $Z\mid L\sim\operatorname{Pois}(L)$, and independent $V\sim\operatorname{Unif}[-1,1]$. Define $W_{\rm cal}=(Z+1+a)e^{\eta V}$ and $W_{\rm test}=(Z+a)e^{\eta V}$. At a unique continuous calibration $\alpha$ quantile $q_\alpha$, coverage tends to
\begin{equation}\label{eq:mixedlimit}
1-F_{\rm test}(q_\alpha)\leq1-\alpha,
\end{equation}
strictly when the coupling crosses $q_\alpha$ with positive probability. No score-band separation is required for this last formula.
\end{proposition}
At $\lambda=6$ and nominal $90\%$, the identical-score limit is $84.8796\%$. The continuous parameters $(\mu,\eta,a)=(.015,.002,1)$ give $79.1749\%$ at nominal $90\%$ and $87.2734\%$ at nominal $95\%$. They are different constructions, not alternative estimates of one population.

\subsection{One deterministic contraction controls all iterations}\label{app:convergence}

For positive vectors let $H(x,y)=\max_h\log(x_h/y_h)-\min_h\log(x_h/y_h)$. Work modulo positive rescaling, equivalently modulo constant shifts of logarithms. Let $G\in\R^{P\times K}$ have positive entries, let $d>0$, and put
\[
\Phi_d(b)=\frac{d}{G^\top[1/(Gb)]},\qquad
b_1=\frac{d}{G^\top\ones},\quad b_{t+1}=\Phi_d(b_t),\quad
p^{(t)}_{ih}=\frac{G_{ih}b_{t,h}}{(Gb_t)_i}.
\]
All divisions are coordinatewise; rescaling any $b_t$ is immaterial. Set $g_-=\min G_{ih}$, $g_+=\max G_{ih}$, $\kappa=1-g_-/g_+$, and $q=\kappa^2<1$.

\begin{lemma}[Contraction with a dimension independent coefficient]\label{lem:contraction}
$H(\Phi_d(b),\Phi_d(c))\leq qH(b,c)$. The map has a unique projective fixed point $b_\infty$, every positive initialization converges projectively to it, and the corresponding probabilities have column sums $Pd_h/\sum_kd_k$.
\end{lemma}
\begin{proof}
For a positive matrix $A$ with entries in $[g_-,g_+]$, the derivative of $z\mapsto\log(Ae^z)$ has stochastic rows
\[
W_{ij}(z)=\frac{A_{ij}e^{z_j}}{\sum_k A_{ik}e^{z_k}}
\ \geq\ \frac{g_-}{g_+}\frac{e^{z_j}}{\sum_ke^{z_k}}.
\]
Every row therefore contains the same probability component of mass $g_-/g_+$. Subtracting this component shows $\operatorname{osc}(Wv)\leq\kappa\operatorname{osc}(v)$. Integrating the derivative along a log-coordinate segment gives $H(Ax,Ay)\leq\kappa H(x,y)$. Reciprocals preserve $H$, and multiplication by the same positive $d$ does not change it. The two positive matrix maps in $\Phi_d$ yield $q=\kappa^2$. The finite dimensional quotient space is complete, so the contraction theorem gives the fixed point and convergence. At a projective fixed point, $\Phi_d(b)=cb$. Writing $s_h=\sum_i p_{ih}$ yields $s_h=d_h/c$. Summing gives $c=(\sum_hd_h)/P$, establishing the prescribed column sums. This is a self-contained version of a classical matrix scaling argument \citep{knight2008}.
\end{proof}

\begin{proposition}[One residual controls every iteration]\label{prop:uniformenvelope}
Let $m=d/\sum_hd_h$, $\bar p_{ih}=G_{ih}m_h/(Gm)_i$, $a_1=H(b_1,m)$, and $r_m=H(\Phi_d(m),m)$. For every $t\geq1$,
\[
H(b_t,m)\leq q^{t-1}a_1+\frac{1-q^{t-1}}{1-q}r_m,
\qquad H(b_\infty,m)\leq \frac{r_m}{1-q}.
\]
Consequently,
\begin{equation}
\sup_{t\in\mathbb N\cup\{\infty\}}\max_{i,h}
\left|\log\frac{p^{(t)}_{ih}}{\bar p_{ih}}\right|
\leq B_G(m):=\max\left\{a_1,\frac{r_m}{1-q}\right\}.
\label{eq:uniformfull}
\end{equation}
\end{proposition}
\begin{proof}
The contraction and the triangle inequality give $H(b_{t+1},m)\leq qH(b_t,m)+r_m$. Sum the geometric recursion. For a row normalized probability, its log ratio is a coordinate of $\log(b_t/m)$ minus the logarithm of a weighted average of $b_t/m$. The latter lies between the extreme coordinate log ratios. Hence its absolute value is at most $H(b_t,m)$. The fixed point version follows either directly or by taking the limit. No probabilistic assumption or finite iteration cap enters this envelope.
\end{proof}

\subsection{Occupancy and quantile convergence}
\label{app:occupancy}
We first prove the identical-score part of Proposition~\ref{thm:flat}. Write $c=(c_1,\ldots,c_K)\sim\operatorname{Multinomial}(n;1/K,\ldots,1/K)$, with $n/K\to\lambda>0$. For every fixed integer $t\geq0$, define
\[
H_K(t)=\frac1K\sum_{h=1}^K\ind\{c_h\leq t\},\qquad
B_K(t)=\frac1n\sum_{h=1}^Kc_h\ind\{c_h\leq t\}.
\]
The marginal binomial count converges to $Z\sim\Poi(\lambda)$. Any fixed pair of different class counts converges jointly to two independent copies of $Z$. The latter follows, for example, by taking the limit of the multinomial probability generating function
$[1+(s-1)/K+(t-1)/K]^n\to\exp\{\lambda(s-1)+\lambda(t-1)\}$.
The variance of the average of any bounded function of the counts is its single coordinate variance divided by $K$ plus $(K-1)/K$ times the two coordinate covariance. The first term vanishes and the second converges to zero. Applying this to indicators and to the bounded function $j\ind\{j\leq t\}$ gives
\begin{equation}
H_K(t)\xrightarrow{p}F_\lambda(t),\qquad
B_K(t)\xrightarrow{p}\lambda^{-1}\E[Z\ind\{Z\leq t\}]=F_\lambda(t-1).
\label{eq:occconv}
\end{equation}
The identity uses $j\Prob(Z=j)=\lambda\Prob(Z=j-1)$.

For the empirical prior score, sort class counts in decreasing order and replicate each count according to its calibration multiplicity. Let $R_K$ be the count at rank $k_\alpha$ in this list. Equivalently, at the ascending lower end, fewer than $n-k_\alpha+1$ observations have count less than $R_K$, while at least that many have count at most $R_K$. Since $(n-k_\alpha)/n\to\alpha$, the strict gaps in \eqref{eq:phase} and \eqref{eq:occconv} imply $\Prob(R_K=r)\to1$. Conditional on $c$, an independently uniform target label is accepted with probability $K^{-1}\sum_h\ind\{c_h\geq R_K\}$. This random fraction converges in probability to $1-F_\lambda(r-1)$. Bounded convergence in probability implies convergence of its expectation. A finite uniform pseudocount, even one varying with $(n,K)$, never changes these exact real arithmetic ranks.

For augmentation, let the hypothetical true target class be $h$. Only the scores of its $c_h$ calibration observations change count level; all other calibration count levels remain unchanged. A union bound and the binomial tail bound imply $\max_hc_h=O_p(\log K)$. Thus the fraction of affected calibration observations is uniformly $o_p(1)$ over $h$. The strict quantile gaps therefore keep the augmented calibration cutoff at $r$ with probability tending to one, uniformly over candidates. The candidate count is now $c_h+1$, so the limiting acceptance event is $Z\geq r-1$.

For the smoothed full conformal rule, let $G_h$ and $E_h$ be the numbers of calibration nonconformities respectively strictly larger than and equal to the target nonconformity under its augmented fit. With an independent $V_h\sim\operatorname{Unif}(0,1)$, set
\begin{equation}
p_h=\frac{G_h+V_h(E_h+1)}{n+1},\qquad C=\{h:p_h>\alpha\}.
\label{eq:smoothed}
\end{equation}
At the true label, exchangeable rank blocks partition $[0,1]$ into intervals of lengths proportional to their multiplicities. Uniform randomization within the selected block makes $p_{Y_*}$ uniform, including ties. Coverage is exactly $1-\alpha$. This is the established smoothed conformal construction.

\subsection{Continuous informative scores}
\label{app:continuous}
We prove Proposition~\ref{thm:continuous} for fixed $\mu\in\R$, $\eta>0$, $a>0$, arbitrary $T_K\in\mathbb N\cup\{\infty\}$, and $n/K\to\lambda\in(0,\infty)$. Perfect original classification additionally requires $\mu>2\eta$. There is one target, so $P=n+1$. For clarity, all limits concern exact real arithmetic.

\paragraph{Uniform approximation of all cycles and their exact limit}
Let $D_h$ be the count of class $h$ in all $P$ labels. With probability tending to one, $\max_h D_h\leq C\log K$ for a fixed sufficiently large $C$. Indeed, for $t>eP/K$, a binomial upper tail is bounded by $(eP/(Kt))^t$; multiply by $K$ and take $t=C\log K$. On this event,
\[
\max_h m_h=O(\log K/K),\qquad \sum_hm_h^2\leq\max_hm_h.
\]
Write $V_{ih}=e^{\eta U_{ih}}$, so $G_{ih}=V_{ih}e^{\mu\ind\{Y_i=h\}}$, and $b_\eta=\E V_{ih}$. Conditional on all labels, the $V_{ih}$ are independent and bounded above and below by positive constants. Each column mean differs from $P^{-1}\sum_i V_{ih}$ by at most a fixed constant times $D_h/P$. Bounded variable concentration and a union bound over $K$ columns give
\begin{equation}
\max_h\left|P^{-1}\sum_iG_{ih}-b_\eta\right|=o_p(1).
\label{eq:columnconc}
\end{equation}
Conditional on labels, each weighted row sum has expectation
$b_\eta[1+(e^\mu-1)m_{Y_i}]$ and sum of squared bounded weights at most a constant times $\sum_hm_h^2$. Hoeffding's inequality and a union bound over $P=O(K)$ rows give
\begin{equation}
\max_i\left|\sum_hm_hG_{ih}-b_\eta\right|=o_p(1).
\label{eq:rowconc}
\end{equation}
For instance, deviations of order $(\log K)^{3/2}/\sqrt K$ have a vanishing union probability on the displayed occupancy event. Conditioning is legitimate: counts depend on labels, not on the independent $U$ arrays.

The concentration statements imply that the initial projective distance $H(b_1,m)$ and residual $H(\Phi_d(m),m)$ are both $o_p(1)$. The contraction coefficient is bounded by the fixed $q_0<1$ from Appendix~\ref{app:convergence}. Proposition~\ref{prop:uniformenvelope} therefore bounds the probability log approximation uniformly over \emph{all} $t\geq1$ and $t=\infty$, proving \eqref{eq:uniformapprox}. This replaces a fixed-length induction. The fully explicit residual argument and finite sample error envelope appear in Appendix~\ref{app:convergence}.

\paragraph{The two limiting score distributions}
Multiply every true probability by the same deterministic number $(n+Ka)b_\eta e^{-\mu}$. By the preceding approximation the result is
\[
(c_{Y_i}+a)e^{\eta U_{i,Y_i}}(1+o_p(1))
\]
uniformly, for calibration and target. Let $V\sim\operatorname{Unif}[-1,1]$ be independent of $Z\sim\Poi(\lambda)$. The calibration empirical distribution converges, at every fixed probability level argument, to the law
\begin{equation}
W_{\mathrm{cal}}=(Z+1+a)e^{\eta V}.
\label{eq:calcontinuous}
\end{equation}
To see this directly, condition on labels. The selected variables $U_{i,Y_i}$ are independent uniforms. The conditional variance of the empirical CDF is at most $1/(4n)$. Its conditional expectation is a count weighted average of the corresponding uniform CDFs. For a fixed argument only finitely many count levels contribute, because $e^{\eta V}$ is bounded below. Equation~\eqref{eq:occconv} therefore identifies its limit. The uniform relative approximation does not change the limit, which has a continuous CDF.

For the target, $c_{Y_*}$ has the marginal binomial count law, is independent of its own selected uniform, and converges to $Z$. Thus its rescaled true probability converges in distribution to
\begin{equation}
W_{\mathrm{test}}=(Z+a)e^{\eta V}.
\label{eq:testcontinuous}
\end{equation}
The target and the empirical calibration CDF need not be independent. Since the latter converges to a deterministic distribution, convergence of its quantile and continuity of the target limit suffice for the coverage calculation.

Under \eqref{eq:bandgap}, the entire count $r$ band
$[(r+a)e^{-\eta},(r+a)e^\eta]$ is separated from every lower and higher count band. By \eqref{eq:phase}, the $\alpha$ quantile of \eqref{eq:calcontinuous} lies strictly inside that band. Define
\[
\rho=\frac{\alpha-F_\lambda(r-2)}{\Prob(Z=r-1)}\in(0,1),\qquad
q_\alpha=(r+a)\exp\{\eta(2\rho-1)\}.
\]
This is the unique $\alpha$ quantile. The ascending calibration probability threshold in \eqref{eq:rank} has rank $n-k_\alpha+1$, so it converges in probability to $q_\alpha$. The target is rejected below this threshold. Its limiting rejection probability is
\[
F_\lambda(r-1)+\rho\Prob(Z=r)
=F_\lambda(r-1)+\frac\lambda r\bigl[\alpha-F_\lambda(r-2)\bigr].
\]
Taking the complement proves \eqref{eq:continuouslimit}. The numerical values in the theorem follow at $r=4$ and $r=3$. They are limiting values, not confidence bounds for the finite experiment.

\paragraph{What changes when there is no continuous variation}
At $\eta=0$, the calibration probability quantile is an atom. Nonstrict inclusion accepts the entire target count $r$ band, giving the identical-score case of Proposition~\ref{thm:flat} rather than \eqref{eq:continuouslimit}. Thus the order of the limits matters: taking $K\to\infty$ for any fixed sufficiently small positive $\eta$ is not equivalent to setting $\eta=0$ first. The continuous formula does not require a large noise magnitude; it requires a nondegenerate ordering within the limiting band. This distinction also explains why comparing LAC and NLL after finite precision tie merging can obscure the mechanism.

\subsection{Uniform count bias at all stopping times and at convergence}
Under the continuous model, $G_{ih}=\exp(\eta U_{ih}+\mu\ind\{Y_i=h\})$, and therefore
\[
q\leq q_0:=\bigl(1-e^{-(|\mu|+2\eta)}\bigr)^2<1.
\]
For fixed $\mu,\eta$, this coefficient is bounded away from one independently of $n,K$, though it can be conservative. The column and weighted row concentration statements in Appendix~\ref{app:continuous} imply
\[
a_1=\operatorname{osc}\log(G^\top\ones)=o_p(1),\qquad
r_m=\operatorname{osc}\log\{G^\top[1/(Gm)]\}=o_p(1).
\]
Indeed $(Gm)_i=b_\eta(1+o_p(1))$ uniformly, and each column mean is $b_\eta(1+o_p(1))$. Substituting these into the second expression gives a common factor $P$ times $1+o_p(1)$. Proposition~\ref{prop:uniformenvelope} then gives
\begin{equation}
\sup_{t\in\mathbb N\cup\{\infty\}}\max_{i,h}
\left|\log\frac{p^{(t)}_{ih}}{m_hG_{ih}/b_\eta}\right|=o_p(1).
\label{eq:uniformapprox}
\end{equation}
This bound controls an unbounded number of iterations through one residual.

After multiplying probabilities by $(n+Ka)b_\eta e^{-\mu}$, the same proxy scores $(c_{Y_i}+a)e^{\eta U_{i,Y_i}}$ approximate \emph{every} iteration and the exact limit. Their empirical calibration distribution and target distribution converge as proved in Appendices~\ref{app:continuous} and~\ref{app:mixed}. At a unique continuous calibration quantile, uniform log-score approximation also gives uniform quantile approximation. Outside an event on which the target proxy is arbitrarily close to that limiting quantile, every iteration has the same acceptance decision as the proxy. The probability of this exceptional event tends to zero by continuity of the target limit.

It follows that all continuous cases of Proposition~\ref{thm:continuous} hold for \emph{every measurable} $T_K\in\mathbb N\cup\{\infty\}$, including stopping rules that depend on the pooled features and calibration labels. No deterministic upper bound on $T_K$ is required. This does not justify arbitrary stopping in the \emph{valid augmented procedure}: its ideal score definition must still be permutation equivariant, or its numerical output must be a certified outer set of that ideal procedure. The exact converged positive-kernel fit is equivariant because it is unique. The result fixes the model parameters and assumes bounded class intensities; it does not assert uniformity for growing score contrast, arbitrary kernels, or changing catalogs outside the stated triangular array.

\subsection{Mixed Poisson extension}
\label{app:mixed}
We prove the unequal-frequency part of the complementary proposition. The empirical quantile cannot be inferred solely from marginal self-count dominance.

For the transport specialization, the arrays of class probabilities are deterministic. Write $\lambda_K=n/K$, $\lambda_{K,h}=n\pi_{K,h}$, and assume $0\leq\lambda_{K,h}\leq L_{\max}<\infty$. Weak convergence of the empirical intensity laws on this compact interval implies convergence of their means to $\lambda$, so $H^*(d\ell)=\ell H(d\ell)/\lambda$ is a probability law. Classes of zero probability have zero weight in $H^*$; a positive pseudocount keeps their transport columns defined.

\paragraph{Normalization still concentrates.}
Each full sample class count is binomial with mean at most $L_{\max}(1+1/n)$. The same tail bound and union argument as Appendix~\ref{app:continuous} give $\max_hD_h=O_p(\log K)$. Consequently $\max_hm_h=O_p(\log K/K)$. Conditional on the labels, every noise variable remains independent. Column means still differ from their noise only means by at most a fixed constant times $\max_hD_h/(n+1)$. Weighted row means differ from $b_\eta$ by at most a constant times $\max_hm_h$, and their squared weight sums vanish. Thus \eqref{eq:columnconc}, \eqref{eq:rowconc}, and the uniform contraction argument of Appendix~\ref{app:convergence} hold unchanged for all iterations and the exact limit. The rescaled true probabilities have uniform relative approximation $(c_{Y_i}+a)e^{\eta U_{i,Y_i}}$.

\paragraph{Calibration distribution.}
For fixed $x>0$, specialize the generic notation to
\[
g_x(j)=\Prob\{(j+a)e^{\eta V}\leq x\}.
\]
This function vanishes when $j>x e^\eta-a$, so $j g_x(j)$ is bounded for this fixed $x$. Conditional on the labels, the independent noise variables give calibration empirical CDF variance at most $1/(4n)$ about
\[
T_K(x)=\frac1n\sum_h c_h g_x(c_h).
\]
Changing one calibration label changes at most two counts, hence changes $T_K(x)$ by at most a fixed $x$ dependent constant divided by $n$. The independent label bounded differences variance bound gives $\operatorname{Var}(T_K(x))=O(1/n)$. Its expectation equals
\[
\sum_h\pi_{K,h}\E\left[g_x\{1+\operatorname{Bin}(n-1,\pi_{K,h})\}\right].
\]
For completeness, a binomial count of $N$ Bernoulli variables with success probability $p$ can be coupled to a sum of $N$ independent $\Poi(p)$ variables with disagreement at most $Np^2$. This follows by coupling one Bernoulli to $\ind\{\Poi(p)>0\}$, then charging both the difference $p-(1-e^{-p})\leq p^2/2$ and the probability of at least two Poisson events, also at most $p^2/2$. Changing Poisson mean $(n-1)p$ to $np$ adds error at most $p$. The intensity bound therefore makes the binomial to $\Poi(\lambda_{K,h})$ error uniformly $O(1/n)$.

It follows that the expectation converges to
\[
\frac1\lambda\int \ell\,\E[g_x\{1+\Poi(\ell)\}]\,H(d\ell),
\]
because the integrand is bounded and continuous on $[0,L_{\max}]$. This is exactly $F_{\rm cal}(x)$ in Proposition~\ref{thm:mixed}. Combining conditional noise concentration and label concentration proves empirical CDF convergence at each $x$. The uniform relative score approximation leaves this limit unchanged because $F_{\rm cal}$ is continuous.

\paragraph{Target distribution and the quantile.}
An independent target chooses class $h$ with probability $\pi_{K,h}$. Conditional on that choice, its calibration count is $\operatorname{Bin}(n,\pi_{K,h})$, and its selected noise variable is independent. The same uniform coupling and empirical intensity convergence give target limit $F_{\rm test}$. Uniqueness of the continuous calibration $\alpha$ quantile implies convergence of the ascending probability threshold at rank $n-k_\alpha+1$ to $q_\alpha$. Independence between that threshold and the target is unnecessary: convergence of the former to a constant and continuity of the latter's limiting law determine the limiting acceptance probability. This proves \eqref{eq:mixedlimit}.

Under the common $L,Z,V$ coupling, $W_{\rm cal}=W_{\rm test}+e^{\eta V}$. Therefore $F_{\rm test}(x)\geq F_{\rm cal}(x)$ for every $x$. Their difference at $q_\alpha$ is exactly the crossing probability, with endpoints irrelevant by continuity. This proves both the nonpositive coverage error and its stated strictness condition. The result needs neither equal class frequencies nor disjoint bands. It does require bounded intensities and a unique quantile; no claim is made about arbitrary unbounded intensity limits or quantiles on a flat CDF interval.

\section{Rank erosion, finite-cycle scope, and numerical certificates}\label{app:numerics}
\subsection{One cycle erosion and linear augmentation}\label{app:erosion}
For the empirical prior $d_h=c_h+a$, the general update (\ref{eq:generalupdate}) has $\rho_h=1+1/d_h$. For a candidate with original probability $t=p_*(h)$, an other-label calibration row has $v_i=p_i(Y_i)$ and $u_i=p_i(h)$. It ceases to be strictly better exactly when
\[
1<v_i/t\leq\frac{(d_h+u_i)(d_h+1)}{d_h(d_h+t)}\leq(1+1/d_h)^2.
\]
If $b_h$ comparisons were originally better and $r_h$ cross this window, the new count is $b_h-r_h$. An excluded candidate enters precisely when $r_h\geq b_h-k_\alpha+1$; same-label comparisons never cross.

For a labeled exchangeable bag of $P=n+1$ observations, the true augmented counts are the same for every target role. Let $A_i$ and $E_i$ denote augmented and ordinary inclusion for role $i$. One cycle gives $E_i\leq A_i$, and consequently
\[
\frac1P\sum_i E_i=\frac1P\sum_i A_i-\frac1P\sum_i A_i(1-E_i),\qquad
\Pr\{Y_*\in C_E\}=\mathbb E B_\alpha-\mathbb E L_\alpha.
\]
Sorting the common augmented true score vector gives $B_\alpha\geq k_\alpha/P$, with equality if scores are distinct. Exchangeability identifies a designated target's expected coverage with the average over roles. The roles are dependent, so the bag is the statistical unit.

For a one-cycle kernel, put $A_{ih}=G_{ih}/\sum_jG_{jh}$ and $t_i=\sum_h(c_h+a)A_{ih}$. Then
\begin{equation}\label{eq:closedform}
p_i^{(h)}(j)=\frac{(c_j+a+\ind\{j=h\})A_{ij}}{t_i+A_{ih}}.
\end{equation}
The $nK$ calibration-hypothesis true scores and $MK$ candidate scores cost $O((n+M)K)$ to construct; fixed-level quantiles use linear selection. This is not a one-cycle formula for a converged fit.

\subsection{A three cycle counterexample to nesting}\label{app:threecycle}
Take calibration labels $(2,2,2,2,2,1,1,2,0)$, $a=1$, $\alpha=.3$, and the last row as the target in
\[
G=\begin{pmatrix}
836&229&485\\453&947&493\\407&127&863\\78&505&199\\799&464&922\\
117&116&821\\697&617&819\\982&761&854\\637&441&425\\504&998&69
\end{pmatrix}.
\]
Three complete rational normalization cycles give empirical set $\{0,1\}$ and full augmented set $\{1\}$. Every entry is positive. This exact witness rules out extending one-cycle nesting to repeated transport. It does not contradict the converged augmented rank guarantee.

\subsection{Prior sensitivity and a local limiting-score certificate}\label{app:newsensitivity}
For the positive-kernel map $\Phi_d$ above, let $\Delta=H(d',d)$ and $q<1$ be its contraction coefficient. The two prior sequences obey $H(b'_{t+1},b_{t+1})\leq\Delta+qH(b'_t,b_t)$, initially $H(b'_1,b_1)=\Delta$. Thus
\begin{equation}\label{eq:sensitivity}
\max_{i,h}|\log p_i'^{(T)}(h)-\log p_i^{(T)}(h)|\leq\theta_T\Delta,\quad
\theta_T=\frac{1-q^T}{1-q},\quad\theta_\infty=\frac1{1-q}.
\end{equation}
Row normalization converts projective distance into the probability bound. Order statistics are 1-Lipschitz in the maximum norm. For candidate augmentation, $\Delta=\log(1+1/d_h)$, so a query/threshold margin exceeding $2\theta_T\Delta$ resolves the corresponding reference comparison. The earlier timings used the weaker $T\Delta$ bound and are not reassigned to the stronger estimate. This specializes established stability screening \citep{ndiaye2022}.

\subsection{A local certificate for the exact converged score}
For any positive proposed multiplier $b$, write $p_{ih}=G_{ih}b_h/(Gb)_i$, $s_h=\sum_ip_{ih}$, and
\[
J_{hk}=\frac{\sum_i p_{ih}p_{ik}}{s_h},\qquad
\zeta=\sum_k\min_hJ_{hk}>0,\qquad
\delta_b=\operatorname{osc}_h\log(d_h/s_h).
\]
The rows of $J$ sum to one. It is the derivative of $z\mapsto\log\Phi_d(e^z)$ at $z=\log b$. Any certified lower bound $\underline\zeta\leq\zeta$ and upper bound $\overline\delta\geq\delta_b$ may replace the exact values below.

\begin{proposition}[Residual to converged rank certificate]\label{prop:localcertappendix}
If $r=2\overline\delta/\underline\zeta\leq1/6$, then
\[
H(b,b_\infty)\leq r,\qquad
\max_{i,h}|\log p_{ih}-\log p^{(\infty)}_{ih}|\leq r.
\]
Candidate and calibration score intervals obtained from this bound give inner and outer prediction sets. Including every unresolved candidate in the outer set preserves the ideal full-conformal coverage guarantee under the reference rank argument in Section~\ref{sec:setup}'s sampling and equivariance assumptions, even when the computation is stopped early.
\end{proposition}
\begin{proof}
On a projective ball of radius $r$ around $b$, every row probability changes multiplicatively by a factor in $[e^{-r},e^r]$. Each entry of the log-map derivative therefore obeys $J_{hk}(\widetilde b)\geq e^{-3r}J_{hk}(b)$. Its rows share a component of mass at least $e^{-3r}\underline\zeta$. Integrating the derivative on the convex log-coordinate ball gives local contraction coefficient at most $1-e^{-3r}\underline\zeta$. Since $r\leq1/6$, $e^{-3r}\geq1-3r\geq1/2$, and
\[
\overline\delta+(1-e^{-3r}\underline\zeta)r\leq r.
\]
Hence the map sends the closed ball to itself and is a contraction there. Its fixed point is the unique positive-kernel projective fixed point, proving the distance bound. Row normalization converts projective distance to the stated log-probability bound.

For candidate $h$, let $[\ell_i,u_i]$ enclose each calibration true probability and $[\ell_*,u_*]$ enclose its query probability in the \emph{same augmented fit}. The number of definitely better calibration probabilities is $B^-_h=\#\{i:\ell_i>u_*\}$, while $B^+_h=\#\{i:u_i>\ell_*\}$ bounds the number that could be better. Known identical feature rows with the same class have exact ties and are removed from both counts. The ideal candidate is certainly included when $B^+_h<k_\alpha$ and certainly excluded when $B^-_h\geq k_\alpha$. Otherwise retain it in the outer set. This set contains the exact converged full-conformal set pointwise, so its marginal coverage is at least that of the ideal set. This is a stability-based conformal outer-set construction, not a new exchangeability principle \citep{ndiaye2022}.
\end{proof}

The certificate concerns \emph{ranks at convergence}, not just a small marginal residual. A residual without a stability factor need not imply a small score error. Conversely, a strictly resolved rank can be certified before extremely accurate convergence of every matrix entry. The implementation retains unresolved cases rather than counting iteration-cap termination as convergence. Full prediction sets require candidate hypotheses; the large convergence study below evaluates only the true hypothesis for coverage and makes no new set-size claim.

\subsection{Outward arithmetic and the numerical contract}
The new experiment interprets the positive binary64 matrix $G=\operatorname{fl}(\exp(L-\max L))$ as the exact input kernel. It does not equate a rounded exponential to the exact real exponential. This feature-only map is permutation equivariant, and the mathematical theorems apply to every positive kernel. Range checks reject zero, subnormal intermediate products, and nonfinite arithmetic rather than silently changing the kernel.

Here is the explicit rounding contract. Let $u=2^{-53}$, $\gamma_m=mu/(1-mu)$, and assume positive dot products obey the standard binary64 relative error bound $\gamma_m$ without underflow or overflow. If $\widehat p$ is computed by a product, a length-$K$ dot product and a division, its exact value for the proposed binary64 $b$ obeys
\[
c_-\widehat p_{ih}\leq p_{ih}\leq c_+\widehat p_{ih},\quad
c_- =\frac{1-\gamma_K}{(1+u)^2},\quad
c_+ =\frac{1+\gamma_K}{(1-u)^2}.
\]
The code encloses these coefficients with directed adjacent floating values. A length-$P$ sum then gives $s_h^-=c_-\widehat s_h/(1+\gamma_P)$ and $s_h^+=c_+\widehat s_h/(1-\gamma_P)$, rounded outward. From the Gram product of $\widehat p$, a lower matrix bound is
\[
J^-_{hk}=\frac{c_-^2\,\operatorname{fl}(\widehat p^\top\widehat p)_{hk}}{(1+\gamma_P)s_h^+}.
\]
Outward sums of its columnwise minima give $\underline\zeta$. Set $R^+=\max_h(d_h/s_h^-)/\min_h(d_h/s_h^+)$. The inequality $\log R^+\leq R^+-1$ supplies $\overline\delta$ without any logarithm in the certifier. Bounds for the pointwise probability calculation are added to $r$; for a resulting log radius $\epsilon<1$, the interval $[(1-\epsilon)\widehat p,\widehat p/(1-\epsilon)]$ is conservative by $e^{-\epsilon}\geq1-\epsilon$ and $e^\epsilon\leq1/(1-\epsilon)$. Every elementary endpoint operation is rounded outward.

The mathematical enclosure assumes the stated binary64 arithmetic model and successful range checks. Separate 80-decimal scalar calculations check the declared subset of small positive-kernel instances. The preceding proof establishes the general real-arithmetic result.

\section{Empirical design, complete comparisons, and interpretation}\label{app:empirical}
\paragraph{Scope of inference.} The original feature extraction uses nine encoder/dataset pairs, with 45,999 encoded image/encoder pairs. A fixed hash excludes 256 temperature-fitting rows per dataset before every evaluation. Resampling from the remaining finite pool is iid conditional on that pool even when the original images are not independent draws from a world population. This does not estimate uncertainty over different datasets, geographic environments, or pretrained-model training corpora. Official test images were not used in this project's earlier development; their absence from encoder pretraining is not established.

\paragraph{Primary allocation.} The one-cycle and three-cycle bag study has $9\times2\times3\times2\times2\times3=648$ endpoints: cache, temperature, ratio, cycle count, nominal level, and empirical coverage/augmented coverage/net role loss. With independent bag summaries $V_r\in[a,b]$, unbiased sample variance $s^2$, and $R$ repetitions, the empirical Bernstein radius used is
\[
\sqrt{\frac{2s^2\log(4/\delta)}{R}}+\frac{7(b-a)\log(4/\delta)}{3(R-1)},\qquad \delta=.05/648.
\]
Intersect with the declared range. Coverage and one-cycle loss lie in $[0,1]$; net three-cycle loss lies in $[-1,1]$. The family is fixed, not selected after inspecting intervals. Shared observations across endpoints do not invalidate a union bound. Native and fitted temperatures remain distinct; a complete table of every endpoint is in \texttt{results/summary/all\_coverage\_conditions.csv}.

\paragraph{Later convergence allocation.} Each of 27 cells has 1,024 new target episodes from the already known pools. Two priors, four iteration choices, and two levels give 432 intervals with Clopper--Pearson tails \citep{clopper1934} of size $.05/(2\cdot432)$. Guaranteed inclusions determine a lower count; possible inclusions determine an upper count. An unresolved numerical rank would widen the interval. Here all 55,296 limit fits and all target ranks resolve. The tables below retain every exact-limit condition, including larger-budget nondetections.
\begin{table}[p]\centering\small
\caption{Every exact-limit condition at nominal 90\% in the post-confirmation sensitivity. Each cell uses $1,024$ independent target episodes. Coverage and simultaneous $95\%$ intervals are percentages; all intervals use the separate $432$-endpoint allocation and round outward. This table does not replace the original frozen one-cycle results.}
\begin{tabular}{llrll}\toprule Encoder & Data & $n/K$ & Empirical limit & Augmented limit\\\midrule
CLIP B/32 & CIFAR10 & 2 & 75.10 [69.60, 80.09] & 88.77 [84.52, 92.23]\\
CLIP B/32 & CIFAR10 & 6 & 85.64 [81.02, 89.55] & 90.82 [86.88, 93.94]\\
CLIP B/32 & CIFAR10 & 20 & 87.60 [83.20, 91.24] & 89.26 [85.08, 92.64]\\
CLIP B/32 & Food101 & 2 & 82.81 [77.89, 87.07] & 90.14 [86.09, 93.37]\\
CLIP B/32 & Food101 & 6 & 86.82 [82.32, 90.57] & 89.75 [85.64, 93.05]\\
CLIP B/32 & Food101 & 20 & 89.65 [85.53, 92.97] & 90.04 [85.97, 93.29]\\
CLIP B/32 & Aircraft & 2 & 84.96 [80.26, 88.96] & 89.94 [85.86, 93.21]\\
CLIP B/32 & Aircraft & 6 & 88.48 [84.19, 91.98] & 91.11 [87.22, 94.18]\\
CLIP B/32 & Aircraft & 20 & 88.28 [83.97, 91.82] & 88.57 [84.30, 92.06]\\
CLIP L/14 & CIFAR10 & 2 & 64.26 [58.31, 69.92] & 88.87 [84.64, 92.31]\\
CLIP L/14 & CIFAR10 & 6 & 80.27 [75.13, 84.80] & 89.55 [85.41, 92.89]\\
CLIP L/14 & CIFAR10 & 20 & 87.70 [83.31, 91.32] & 89.26 [85.08, 92.64]\\
CLIP L/14 & Food101 & 2 & 69.73 [63.96, 75.10] & 90.62 [86.65, 93.78]\\
CLIP L/14 & Food101 & 6 & 83.89 [79.07, 88.02] & 91.41 [87.56, 94.42]\\
CLIP L/14 & Food101 & 20 & 88.87 [84.64, 92.31] & 90.92 [86.99, 94.02]\\
CLIP L/14 & Aircraft & 2 & 83.89 [79.07, 88.02] & 89.55 [85.41, 92.89]\\
CLIP L/14 & Aircraft & 6 & 86.72 [82.21, 90.48] & 89.26 [85.08, 92.64]\\
CLIP L/14 & Aircraft & 20 & 88.09 [83.75, 91.65] & 88.96 [84.75, 92.39]\\
SigLIP B/16 & CIFAR10 & 2 & 72.07 [66.41, 77.29] & 91.02 [87.10, 94.10]\\
SigLIP B/16 & CIFAR10 & 6 & 82.52 [77.57, 86.81] & 89.75 [85.64, 93.05]\\
SigLIP B/16 & CIFAR10 & 20 & 87.11 [82.65, 90.82] & 88.77 [84.52, 92.23]\\
SigLIP B/16 & Food101 & 2 & 72.36 [66.72, 77.56] & 90.04 [85.97, 93.29]\\
SigLIP B/16 & Food101 & 6 & 83.59 [78.75, 87.76] & 90.53 [86.54, 93.70]\\
SigLIP B/16 & Food101 & 20 & 90.14 [86.09, 93.37] & 92.09 [88.36, 94.97]\\
SigLIP B/16 & Aircraft & 2 & 85.16 [80.47, 89.13] & 88.77 [84.52, 92.23]\\
SigLIP B/16 & Aircraft & 6 & 88.77 [84.52, 92.23] & 90.23 [86.20, 93.45]\\
SigLIP B/16 & Aircraft & 20 & 88.48 [84.19, 91.98] & 89.06 [84.86, 92.48]\\
\bottomrule\end{tabular}\end{table}
\begin{table}[p]\centering\small
\caption{Every exact-limit condition at nominal 95\% in the post-confirmation sensitivity. Each cell uses $1,024$ independent target episodes. Coverage and simultaneous $95\%$ intervals are percentages; all intervals use the separate $432$-endpoint allocation and round outward. This table does not replace the original frozen one-cycle results.}
\begin{tabular}{llrll}\toprule Encoder & Data & $n/K$ & Empirical limit & Augmented limit\\\midrule
CLIP B/32 & CIFAR10 & 2 & 85.55 [80.91, 89.47] & 94.24 [90.92, 96.66]\\
CLIP B/32 & CIFAR10 & 6 & 92.87 [89.28, 95.59] & 95.21 [92.12, 97.39]\\
CLIP B/32 & CIFAR10 & 20 & 93.46 [89.98, 96.05] & 94.24 [90.92, 96.66]\\
CLIP B/32 & Food101 & 2 & 90.33 [86.31, 93.54] & 95.21 [92.12, 97.39]\\
CLIP B/32 & Food101 & 6 & 92.97 [89.39, 95.67] & 94.92 [91.76, 97.17]\\
CLIP B/32 & Food101 & 20 & 94.63 [91.40, 96.95] & 95.21 [92.12, 97.39]\\
CLIP B/32 & Aircraft & 2 & 92.97 [89.39, 95.67] & 95.31 [92.24, 97.46]\\
CLIP B/32 & Aircraft & 6 & 94.92 [91.76, 97.17] & 95.70 [92.73, 97.75]\\
CLIP B/32 & Aircraft & 20 & 92.87 [89.28, 95.59] & 92.87 [89.28, 95.59]\\
CLIP L/14 & CIFAR10 & 2 & 80.76 [75.66, 85.24] & 94.53 [91.28, 96.88]\\
CLIP L/14 & CIFAR10 & 6 & 89.75 [85.64, 93.05] & 94.73 [91.52, 97.03]\\
CLIP L/14 & CIFAR10 & 20 & 92.97 [89.39, 95.67] & 94.73 [91.52, 97.03]\\
CLIP L/14 & Food101 & 2 & 85.74 [81.12, 89.64] & 95.70 [92.73, 97.75]\\
CLIP L/14 & Food101 & 6 & 91.80 [88.01, 94.73] & 95.70 [92.73, 97.75]\\
CLIP L/14 & Food101 & 20 & 94.43 [91.16, 96.80] & 95.41 [92.36, 97.53]\\
CLIP L/14 & Aircraft & 2 & 91.11 [87.22, 94.18] & 94.82 [91.64, 97.10]\\
CLIP L/14 & Aircraft & 6 & 92.77 [89.16, 95.52] & 94.43 [91.16, 96.80]\\
CLIP L/14 & Aircraft & 20 & 94.53 [91.28, 96.88] & 94.92 [91.76, 97.17]\\
SigLIP B/16 & CIFAR10 & 2 & 83.89 [79.07, 88.02] & 96.48 [93.72, 98.30]\\
SigLIP B/16 & CIFAR10 & 6 & 90.14 [86.09, 93.37] & 94.14 [90.80, 96.58]\\
SigLIP B/16 & CIFAR10 & 20 & 92.29 [88.59, 95.13] & 94.04 [90.68, 96.51]\\
SigLIP B/16 & Food101 & 2 & 87.01 [82.54, 90.73] & 95.21 [92.12, 97.39]\\
SigLIP B/16 & Food101 & 6 & 91.60 [87.78, 94.57] & 95.02 [91.88, 97.25]\\
SigLIP B/16 & Food101 & 20 & 94.63 [91.40, 96.95] & 95.51 [92.48, 97.60]\\
SigLIP B/16 & Aircraft & 2 & 91.50 [87.67, 94.50] & 94.24 [90.92, 96.66]\\
SigLIP B/16 & Aircraft & 6 & 94.92 [91.76, 97.17] & 95.80 [92.85, 97.82]\\
SigLIP B/16 & Aircraft & 20 & 94.63 [91.40, 96.95] & 94.82 [91.64, 97.10]\\
\bottomrule\end{tabular}\end{table}

\paragraph{Complete prediction sets and label budgets.} All 13 methods were frozen before the new GPU extraction. Unadapted LAC, randomized APS, and randomized RAPS are included \citep{romano2020,angelopoulos2021}. RAPS uses the specified penalty $.01$ after rank three. Prior fitting and split ridge use half the total labels for fitting and half for calibration; full methods use the same total number. Temperature labels are additional, shared information. Ridge uses unit regularization, no intercept, and temperature $.1$, with no evaluation-label tuning. Full ridge recomputes calibration and query scores for every candidate, not only the query score.
\begin{table}[t]\centering\small
\caption{Complete prediction sets on the three confirmation datasets, native temperature and nominal $90\%$. Each value equally averages the $27$ cache/sample-size cell means, with $128$ complete episodes per cell. These descriptive mixed-task means are not an efficiency significance test. Catalog fraction is averaged separately. $\dagger$: no general coverage guarantee. All $13$ methods use matched total calibration/fitting label allowances.}
\label{tab:confirmsets}
\begin{tabular}{lrrr}\toprule
Method & Coverage (\%) & Mean labels & Catalog (\%)\\\midrule
Unadapted LAC & 89.88 & 7.037 & 10.06\\
Randomized APS & 89.61 & 7.657 & 11.71\\
Randomized RAPS & 89.63 & 8.049 & 12.05\\
Split ridge & 90.12 & 23.778 & 37.74\\
Full conformal ridge & 90.24 & 18.470 & 29.10\\
Empirical prior, $T=1$ $\dagger$ & 89.01 & 5.342 & 8.56\\
Empirical prior, $T=3$ $\dagger$ & 87.70 & 4.767 & 7.71\\
Uniform prior, $T=1$ & 90.13 & 5.748 & 9.19\\
Uniform prior, $T=3$ & 90.04 & 5.254 & 8.56\\
Independent prior, $T=1$ & 90.24 & 6.244 & 10.53\\
Independent prior, $T=3$ & 90.24 & 6.022 & 10.73\\
Augmented prior, $T=1$ & 90.12 & 5.761 & 9.13\\
Augmented prior, $T=3$ & 90.12 & 5.237 & 8.44\\
\bottomrule\end{tabular}\end{table}

The original screening experiment matches all 6,912 unscreened comparisons and reports a modest median timing ratio of 1.47. This is not a universal speedup or complete-set efficiency measurement at exact convergence. Every complete method cell, query-average episode metric, and paired interval is retained in \texttt{results/summary/}; packing membership masks does not turn queries from one episode into independent trials.

\paragraph{Directional comparisons on known pools.} Both raw and column-normalized bases, all exponents $-1,-.5,0,.5,1$, all nine caches, and all three ratios remain in the later experiment. Its 1,080 confidence endpoints use a separate allocation. The two descriptive tables below show every exponent, including overcoverage and larger sets for negative powers. The method's sufficient guarantee is not an empirical claim that it yields the best sets.
\begin{table}[t]
\centering
\caption{New known-pool directional study, column-normalized base, nominal $90\%$. Each entry equally averages all $27$ cache/budget cell means, each using $128$ full episodes. All exponents are shown; these are descriptive comparisons, not an efficiency significance test. Decreasing weights preserve marginal validity but are not smaller than the fixed uniform baseline here.}
\label{tab:directionresults}
\begin{tabular}{rrrrr}\toprule
 & \multicolumn{2}{c}{Empirical count reuse} & \multicolumn{2}{c}{Full augmentation}\\
$\gamma$ & Coverage (\%) & Mean labels & Coverage (\%) & Mean labels\\\midrule
-1 & 91.13 & 6.349 & 90.26 & 5.853\\
-0.5 & 90.77 & 5.965 & 90.32 & 5.728\\
0 & 90.37 & 5.641 & 90.37 & 5.641\\
0.5 & 89.87 & 5.418 & 90.39 & 5.634\\
1 & 89.28 & 5.245 & 90.40 & 5.662\\
\bottomrule\end{tabular}
\end{table}

\begin{table}[t]
\centering
\caption{New known-pool directional study, raw base, nominal $90\%$. Each entry equally averages all $27$ cache/budget cell means, each using $128$ full episodes. All exponents are shown; these are descriptive comparisons, not an efficiency significance test. Decreasing weights preserve marginal validity but are not smaller than the fixed uniform baseline here.}
\label{tab:directionraw}
\begin{tabular}{rrrrr}\toprule
 & \multicolumn{2}{c}{Empirical count reuse} & \multicolumn{2}{c}{Full augmentation}\\
$\gamma$ & Coverage (\%) & Mean labels & Coverage (\%) & Mean labels\\\midrule
-1 & 91.08 & 7.455 & 90.26 & 7.036\\
-0.5 & 90.66 & 7.176 & 90.26 & 6.975\\
0 & 90.26 & 6.959 & 90.26 & 6.959\\
0.5 & 89.74 & 6.786 & 90.27 & 6.979\\
1 & 89.21 & 6.650 & 90.29 & 7.031\\
\bottomrule\end{tabular}
\end{table}

\paragraph{Retrospective effect-size and target-confidence checks.} The present revision transforms every native original and converged empirical coverage interval into a miss-budget ratio $(1-\operatorname{coverage})/\alpha$. This is a deterministic re-expression of the same interval, not another significance test. It also examines all existing target episodes within four prelisted strata: all targets, correct native predictions, correct predictions with normalized score at least $.9$, and correct predictions with normalized score at least $.99$. Both priors, cycles one and limit, both levels, and all 27 cells yield a separate family of 864 Clopper--Pearson intervals \citep{clopper1934}. No minimum eligible-count filter or favorable-stratum selection is used. A zero eligible count receives $[0,1]$.

Conditioning on correctness or confidence defines a diagnostic population for which ordinary marginal conformal prediction does not promise nominal coverage. These intervals describe conditional subgroups; the theorem guarantees marginal coverage. The selected CIFAR10 illustration below is retrospective; all 864 rows are supplied. In the stated CLIP L/14 $n=20$ limiting run, many of the misses occur even when the original top prediction is correct and confidently ranked. That diagnoses the mechanism but does not imply that every high-confidence subgroup has a deficit. Table~\ref{tab:confidencecheck} shows attenuation of the observed deficit as the confidence threshold increases. In the highest stratum, the 198 correct targets with normalized score at least $.99$ have empirical coverage $182/198\approx91.92\%$, above nominal, with simultaneous interval $[81.49\%,97.60\%]$. This small subgroup and its wide interval establish neither undercoverage nor equivalence to $90\%$; the point estimate is above nominal.
\begin{table}[ht]
\centering\small
\begin{tabular}{lrrr}
\toprule
Target subset & $N$ & Empirical (\%) & Augmented (\%)\\
\midrule
All & 1024 & 64.26 [58.05, 70.15] & 88.87 [84.44, 92.44]\\
Correct & 972 & 67.28 [61.00, 73.17] & 91.77 [87.68, 94.89]\\
Correct, score $\geq.9$ & 775 & 76.52 [69.97, 82.30] & 97.16 [93.96, 98.97]\\
Correct, score $\geq.99$ & 198 & 91.92 [81.49, 97.60] & 100.00 [94.85, 100.00]\\
\bottomrule
\end{tabular}
\caption{Retrospective target strata for CLIP L/14, CIFAR10, $n=20$. Clopper--Pearson intervals \citep{clopper1934} use the separate 864-cell allocation, so the all-target interval differs from the original convergence interval. Conditional coverage is diagnostic, not a promised guarantee; all four strata and every other cell remain in the supplied CSV. Endpoints are rounded outward and verified by exact binomial tails. In the highest confidence stratum, the point estimate is above nominal; its 198 targets and wide interval do not establish either a deficit or nominal equivalence.}
\label{tab:confidencecheck}
\end{table}

\paragraph{Native confidence is not calibrated truth probability.} Confidence above is softmax of the stored native logits over the fixed catalog. For SigLIP it is not a native independent sigmoid score, and for any encoder it is not an assertion of probability calibration. It is kept separate from top-one accuracy and conformal coverage. All targets were already present in the saved episodes; no new GPU inference or environmental replication was performed for this analysis.

\section{Audit receipt and preserved evidence}\label{app:verification}
The computational audit checks the implementation against a separate scalar Fraction reference, including duplicate rows, nonmonotone and heterogeneous functions, and deliberately near-tied probabilities. It also verifies the two symbolic identities used in the main proof and recounts the retrospective target strata from stored logits and labels.
\begin{table}[!htbp]\centering\small
\begin{tabular}{lr}\toprule
Check & Count\\\midrule
Exact candidate membership comparisons & 6,480\\
Guard equals ordinary union augmented & 2,160\\
Guarded whole-bag guarantee checks & 108\\
Exact converse roles & 2,100\\
Exact quota-constant substitutions & 60\\
Explicit quantifier boundary fixtures & 3\\
Saved target strata recounted & 864\\
Exact displayed binomial-tail inequalities & 1,343\\
Original effect rows independently joined & 108\\
\bottomrule\end{tabular}
\caption{Internal computational checks. All declared fixtures passed.}
\end{table}
The worked example adds six exact probability values and two exact ranks. The guard used 676 exact-fallback comparisons in its test schedule; no tolerance changes the final membership. Four malformed input cases were rejected as specified. All 27,648 reused target identities were checked. The smallest distance of a cached native confidence from either subgroup threshold was 4.6e-07, avoiding a threshold ambiguity in this analysis.

A complete second execution reproduced the scientific contents of 15 files, including 513 saved arrays containing 94,120 entries and 1,848 CSV rows. Only execution time and the recorded UTC timestamp were excluded; input and source bindings remained checked. The original scientific result, protocol, and runner files were separately compared byte for byte against the supplied source archive.

The supplementary checks use finite and symbolic calculations to test the implementations against the stated formulas. The universal arguments are proved in the main text and supporting sections.

The artifact preserves the prior 60 page manuscript in \texttt{history/before\_reviewer\_revision/Sharp\_Boundary\_Record.pdf}, the original GPU schedule and outcomes, subsequent study protocols, and code. The archived study protocols and empirical data remain unchanged, with frozen studies and later analyses identified separately. Earlier floating tie diagnostics and the complete corrected rerun remain preserved. The historical release-order reconstruction gave $1.9066\%$ median size reduction and $.00361704$ mean labels saved, versus the earlier report $1.93\%$ and $.00379$; that small unresolved provenance difference is separate from all new full-catalog results.

\section{Efficiency witnesses and expanded selection study}\label{app:utility}
Theorem~\ref{thm:efficiency} has its complete proof in the main text. The program \texttt{utility/exact\_efficiency.py} enumerates all $2^9$ ordered calibration label sequences and both target labels for each law using rational arithmetic. It also supplies all ten count cases and exactly sums their binomial probabilities. For the helpful law at $p=4/5$, the common coverage is $9503481/9765625$, and the two expected sizes are $11194462/9765625$ and $10014814/9765625$. For the harmful law they are respectively $9765621/9765625$ (both coverage and uniform size) and $15532789/9765625$ (count-rule size). These are counterexamples to universal efficiency orderings, not estimated natural performance or claims about an oracle over all class-specific constant weights.

\paragraph{Candidate specification.} Every class weight is fixed conditional on the fitting pool. The 55 live rules are: uniform; $(c+a)^\gamma$ for $a\in\{1,4\}$ and $\gamma\in\{-.25,-.5,-1,-2\}$; $\max(c+1,m)^\gamma$ for $m\in\{2,4,8,16,32\}$ and $\gamma\in\{-.5,-1\}$; a weight of one at $c\leq t$ and $q$ otherwise for $t\in\{0,1,2,5\}$ and $q\in\{.25,.5,.75\}$; and $1-\ell+\ell a/(c+a)$ for $a\in\{1,4,16\}$, $\ell\in\{.25,.5,.75\}$. The fifteen heterogeneous rules comprise six class-specific power vectors (two negative exponents, targeting classes above median fitting confidence, above median fitted frequency, or continuously scaled by fitting confidence), six class-specific steps at $c\leq b n\widehat\pi_h$ for $b\in\{.5,1,2\}$ and $q\in\{.5,.8\}$, and three negative powers multiplied by reciprocal independently estimated class confidence. Shrinkage in the fitting confidence and class frequency estimates keeps all values positive. All formulas and the deterministic enumeration order are in \texttt{utility/sweep.py}.

\paragraph{Selection and information.} The fixed 256 fitting IDs are the original temperature-fit exclusions, selected by the unchanged hash rule. Per class, $\widehat\pi_h=(N_h+1)/(256+K)$, and confidence is a one-observation shrinkage estimate of the native true-label score toward the overall fitting mean. All candidates, including their thresholds and class-specific exponents, depend only on this fitting pool. The 32 selection episodes resample it; they are not fresh training observations or independent draws from the evaluation pool. A selected function is fixed independently of final calibration labels, as required by Theorem~\ref{thm:sharpcommon}. For a random fitting sample under an iid superpopulation interpretation, condition on that sample before applying the guarantee. Here the experiment is explicitly conditional on fixed disjoint empirical fitting and evaluation pools.

\paragraph{All outcomes and intervals.} Each of 54 cache/budget/base combinations has 512 evaluation episodes and both levels. All 55 live rules and their 55 matching independently frozen rules are evaluated. Files record 11,880 rule-level coverage/size summaries, 648 independently selected global/family winners, and 1,296 paired comparisons against uniform or matching frozen rules. Queries from one episode share calibration, so each query-average episode is the statistical unit. For differences in size divided by $K$, the empirical Bernstein bound uses range $[-1,1]$ and $\delta=.05/1296$ across every selected comparison. No comparison excludes zero in this family. The intervals are intentionally conservative; lack of significance proves neither equivalence nor optimality. The selected raw mean differences versus uniform at 90/95\% are $-.3885/-.5447$ labels; the corresponding differences versus the same rule with independently frozen weights are $+.1203/+.1978$. These descriptive differences prevent attributing all adjustment gains to live count information.

\paragraph{Data provenance.} This retrospective analysis uses the previously studied caches. All indices, candidate formulas, selection choices saved before evaluation, per-episode metrics, per-cell summaries, source/input hashes, and execution records are supplied. The main sweep took 173.30 seconds on four CPU workers; it performed no GPU extraction. The original studies remain unchanged. The separately executed native benchmark is reported in Appendix~\ref{app:native}.

\section{General score criterion: reference choices and precise scope}\label{app:general}
The main text contains the complete transfer proof and both common-map converses. This appendix details the reference choices, implementation, and scope.

\paragraph{All-count versus leave-self-out augmentation.}
For hypothetical query label $h$, the all-count reference scores every row with $D=c+e_h$. In the leave-self-out reference, row $i$ instead uses $D-e_{Y_i}$ and the query uses $D-e_h=c$. Both are ordinary full conformal constructions: at the true hypothesis each row's score is a permutation-equivariant function of the labeled bag, treating every role by the same rule. Only the second reference makes the ordinary query score exactly equal to its reference score. This is why pointwise monotonicity in the count is sufficient even if ties appear or score order changes within a class.

The transfer condition needs only counts of total $n$ when proving validity at that $n$. The sufficiency statement permits known sample-size constants and functions selected using independent fitting information. The common separable iff fixes one map on the whole mark/count domain across sample sizes. It does not infer a fixed-$n$ necessity result from a downward step at an unreachable count.

\paragraph{A monotone map separating the two references.}
Let $K=2$, $n=1$, $\alpha=1/2$, and let both the calibration and target labels be 1. Their scalar base scores are 0 and 1. On $s\in[0,1]$, define
\[
\phi(s,c)=\begin{cases}s,&c\leq1,\\1,&c\geq2.\end{cases}
\]
This map is nondecreasing in both arguments. Ordinary calibration uses count 1 and excludes the target because $0<1$. All-count augmentation uses count 2 and includes the tied target. Hence $C_A\nsubseteq C_E$ on this sample. The leave-self-out reference uses count 1 for both rows and equals the ordinary rule; its marginal guarantee is unaffected.

\paragraph{Feature-dependent additive penalties.}
For \eqref{eq:additive}, the calibration-reference difference for a row with $Y_i=j\ne h$ is exactly
\[
S_i(j;c)-\Psi_j(V_i(j),c+e_h-e_j)
=\frac{\lambda w_i(j)}{n+Ka}\geq0.
\]
For $j=h$ it is zero. A fixed rank penalty or another nonnegative function of a label-free base can therefore be used, even when different rows have different penalties. If auxiliary randomness is used, it must be independent and treated as exchangeable row marks, or the score map must separately meet the equivariance assumption. A hyperparameter optimized on final calibration outcomes is not made valid by this calculation. Applying the corollary to a score of the soft TACP form does not assert that every published variant uses this information protocol.

\paragraph{Why the scalar target alone cannot subsume normalization.}
Fix the candidate's own base value and count in \eqref{eq:generalweight}, then vary a different class count with a nonconstant weight. Its denominator, and generally its candidate probability, change. A scalar map of only the original candidate score and its own count cannot represent this dependency in general. The vector-count transfer criterion handles this coupling. The criterion is sufficient for coupled scores; necessity is established for the common-map specializations.

\paragraph{Implementation and finite checks.}
For separable scores, the routine uses ordinary calibration true scores and their values when their own count is reduced by one. For candidate $h$, only calibration rows whose label is not $h$ use the reduced score; the query retains its ordinary score. One order statistic per candidate produces the exact mathematical reference. The floating API retains inclusive represented ties and rejects malformed or nonfinite inputs; it does not certify numerical errors in an arbitrary user-supplied score evaluator. The ordinary/reference union is available when monotonicity is not assumed.

The recorded fixtures contain 1,944 whole-bag guarantee checks, 19,008 true-role comparisons, 54,720 exact candidate containment comparisons, 210 converse roles, and 1,080 comparisons of the floating API with an exact scalar reference. They cover additive, row-weighted, clipped, stepwise and base-order-reversing score maps, as well as normalized probabilities. All declared checks pass. These finite computations test signs and tie handling; the main-text arguments prove the universal claims. The original natural data, outcomes, and confidence allocations are unchanged.

\clearpage
\section{Completed native Conf-OT benchmark: scope and full results}\label{app:native}
\paragraph{Original code and inputs.}
The author executed the released \texttt{jusiro/CLIP-Conformal} extractor and scientific driver at commit \texttt{1a29cedd65a780696e40e2e099dd23d116d9ca39}. Of the 51 tracked source files in the before/after receipts, only the documented dataset-root constant changes. The returned launcher is byte-identical to the distributed version. The run covers the original default two-dataset subset and unadapted comparison. The execution records specify dataset configurations, dependency versions, and output hashes.

DTD \citep{cimpoi2014} uses the repository's Zhou test list: 1,692 images from 47 classes, each with 36 images. This is not the standard 40-test-images-per-class DTD split. Aircraft \citep{maji2013} uses its 100-variant test list: 67 classes with 33 images and 33 with 34, totaling 3,333. Original prompts are \texttt{\{\} texture} and \texttt{a photo of a \{\}, a type of aircraft.} The original CLIP B/16 preprocessing is retained, without the earlier custom study's attribution-strip removal. Both logit arrays match the manifest label order and cache hashes.

\paragraph{Settings and completion.}
The original driver uses \texttt{epsilon=1.0}, \texttt{ot\_iters=3}, \texttt{observed\_marginal=true}, \texttt{p=0.5}, and \texttt{seeds=20}. Conf-OT uses calibration counts without pseudocounts. Original LAC, randomized APS, and randomized RAPS are retained; RAPS has upstream defaults $\lambda=.001$ and $k_{\rm reg}=1$, not the custom study's different values. All 12 commands complete, supplying 24 cells with 20 finite split records each and no missing benchmark outcomes. Interpretations for either result were specified before execution. Native TIM, TransCLIP, FCA and SCA-T are not reproduced.

\paragraph{Fixed quotas and their limits.}
For a class with $N_h$ images, the released splitter fixes $c_h=\max\{1,\lfloor N_h/2\rfloor\}$ before shuffling within-class indices. DTD has 846 calibration and 846 test images, exactly 18 of each per class; its empirical prior is identically uniform. Aircraft has 1,633 calibration and 1,700 test images: calibration counts are 16 or 17, and every class has 17 test images. Its calibration/test class-mixture total variation distance is .01354. Fixed counts remove this particular random self-count dependence, not every possible obstacle to pooled-rank validity under stratification.

The iid analysis permits a calibration observation's own label to change its scoring count relative to a target. Fixed quotas prevent that perturbation in this benchmark. When counts vary across target roles, or fitted choices use calibration outcomes beyond the quotas, this invariance argument no longer applies; failure does not follow automatically. Model, sample size, feature pool, pseudocount and preprocessing also differ from the custom study. The comparison therefore does not isolate stratification as the sole causal explanation.

\paragraph{Reporting.}
The following tables retain every cell from the released detailed arrays. Medians, means and observed ranges over 20 splits are descriptive, not confidence intervals. Shared images and advancing sampling/random-score streams do not create independent new populations. No pooled binomial test is applied, a below-nominal median is not called a proved coverage violation, and size reductions are not a comparison at matched true coverage. Every original outcome remains available.
\clearpage
\begin{table}[!ht]
\centering\small
\caption{All native benchmark cells at nominal 90\%. Coverage entries are percentages; the range is the minimum and maximum observed split value, not a confidence interval. Size is the median split mean number of labels. Each row has 20 repetitions.}\label{tab:native-90}
\begin{tabular}{lllrrrr}
\toprule
Data & Score & Method & Median & Mean & Range & Size\\
\midrule
DTD & LAC & Unadapted & 89.95 & 89.96 & 88.06--91.37 & 10.866\\
DTD & LAC & Conf-OT & 89.48 & 89.83 & 88.42--91.61 & 8.632\\
DTD & APS & Unadapted & 90.25 & 90.15 & 88.53--92.43 & 12.523\\
DTD & APS & Conf-OT & 89.89 & 89.89 & 88.18--91.84 & 10.026\\
DTD & RAPS & Unadapted & 90.25 & 90.11 & 88.30--92.20 & 12.243\\
DTD & RAPS & Conf-OT & 89.95 & 89.93 & 88.18--91.84 & 9.869\\
Aircraft & LAC & Unadapted & 89.76 & 89.82 & 88.41--91.18 & 18.160\\
Aircraft & LAC & Conf-OT & 89.97 & 89.99 & 88.53--91.12 & 14.038\\
Aircraft & APS & Unadapted & 89.91 & 89.84 & 88.82--91.47 & 17.992\\
Aircraft & APS & Conf-OT & 89.85 & 90.00 & 88.41--91.94 & 14.823\\
Aircraft & RAPS & Unadapted & 89.94 & 89.87 & 88.24--91.65 & 17.986\\
Aircraft & RAPS & Conf-OT & 89.88 & 89.98 & 88.47--91.53 & 14.540\\
\bottomrule
\end{tabular}
\end{table}
\begin{table}[!ht]
\centering\small
\caption{All native benchmark cells at nominal 95\%. Coverage entries are percentages; the range is the minimum and maximum observed split value, not a confidence interval. Size is the median split mean number of labels. Each row has 20 repetitions.}\label{tab:native-95}
\begin{tabular}{lllrrrr}
\toprule
Data & Score & Method & Median & Mean & Range & Size\\
\midrule
DTD & LAC & Unadapted & 95.27 & 95.20 & 92.20--96.57 & 17.291\\
DTD & LAC & Conf-OT & 94.33 & 94.65 & 93.38--96.81 & 12.471\\
DTD & APS & Unadapted & 95.45 & 95.18 & 91.73--96.22 & 18.823\\
DTD & APS & Conf-OT & 94.62 & 94.78 & 93.03--96.22 & 13.752\\
DTD & RAPS & Unadapted & 95.27 & 95.19 & 91.96--96.45 & 17.599\\
DTD & RAPS & Conf-OT & 94.62 & 94.65 & 93.14--96.34 & 13.522\\
Aircraft & LAC & Unadapted & 95.12 & 94.98 & 94.00--95.82 & 28.593\\
Aircraft & LAC & Conf-OT & 94.82 & 94.91 & 93.65--96.35 & 20.507\\
Aircraft & APS & Unadapted & 94.97 & 94.94 & 93.94--95.82 & 29.018\\
Aircraft & APS & Conf-OT & 95.29 & 95.15 & 93.71--96.65 & 21.615\\
Aircraft & RAPS & Unadapted & 95.00 & 94.94 & 93.41--95.88 & 30.245\\
Aircraft & RAPS & Conf-OT & 95.12 & 95.05 & 93.82--96.76 & 21.311\\
\bottomrule
\end{tabular}
\end{table}

\paragraph{Resources and post-return audit.}
The author used one NVIDIA H100 80GB HBM3. Recorded elapsed time is 1,495.35 seconds, including wrapper setup, assets, extraction and benchmark jobs, not GPU busy time. Versions are Python 3.10.12, PyTorch 2.0.1/CUDA 11.8, torchvision 0.15.2 and TorchCP 0.1.3; the full dependency freeze and pinned CLIP revision are supplied. The separate audit reconstructs all 480 split records and 24 summaries from detailed arrays, checks exported CSVs and rounded spreadsheet rows, and verifies input/output/protocol bindings. This audit operates on the saved arrays; raw-image inference is rerun through the original notebook. \texttt{python reproduce.py --native-audit} rebuilds these summaries without a GPU or network. The returned archive includes original outputs, both caches, manifests, all commands, logs and receipts; the original notebook remains available for rerunning the subset.
\clearpage


\begin{thebibliography}{99}
\bibitem[Angelopoulos et~al.(2021)]{angelopoulos2021}
Anastasios N. Angelopoulos, Stephen Bates, Jitendra Malik, and Michael I. Jordan.
Uncertainty Sets for Image Classifiers using Conformal Prediction.
In \emph{International Conference on Learning Representations}, 2021.
\url{https://arxiv.org/abs/2009.14193}.

\bibitem[Bossard et~al.(2014)]{bossard2014}
Lukas Bossard, Matthieu Guillaumin, and Luc Van Gool. Food-101: Mining Discriminative Components with Random Forests. In \emph{European Conference on Computer Vision}, 2014. \url{https://data.vision.ee.ethz.ch/cvl/datasets_extra/food-101/}.

\bibitem[Cimpoi et~al.(2014)]{cimpoi2014}
Mircea Cimpoi, Subhransu Maji, Iasonas Kokkinos, Sammy Mohamed, and Andrea Vedaldi.
Describing Textures in the Wild.
In \emph{IEEE/CVF Conference on Computer Vision and Pattern Recognition}, 2014.
\url{https://www.robots.ox.ac.uk/~vgg/data/dtd/}.

\bibitem[Clopper and Pearson(1934)]{clopper1934}
C.~J. Clopper and E.~S. Pearson.
The use of confidence or fiducial limits illustrated in the case of the binomial.
\emph{Biometrika}, 26(4):404--413, 1934. \url{https://doi.org/10.1093/biomet/26.4.404}.

\bibitem[Ding et~al.(2023)]{ding2023}
Tiffany Ding, Anastasios N. Angelopoulos, Stephen Bates, Michael I. Jordan, and Ryan J. Tibshirani.
Class-Conditional Conformal Prediction with Many Classes.
In \emph{Advances in Neural Information Processing Systems}, volume 36, 2023.
\url{https://proceedings.neurips.cc/paper_files/paper/2023/hash/cb931eddd563f8d473c355518ce8601c-Abstract-Conference.html}.

\bibitem[Gazin et~al.(2024)]{gazin2024}
Ulysse Gazin, Gilles Blanchard, and Etienne Roquain.
Transductive conformal inference with adaptive scores.
In \emph{Proceedings of AISTATS}, volume 238, pp. 1504--1512, 2024.
\url{https://proceedings.mlr.press/v238/gazin24a.html}.

\bibitem[Helber et~al.(2019)]{helber2019}
Patrick Helber, Benjamin Bischke, Andreas Dengel, and Damian Borth.
EuroSAT: A novel dataset and deep learning benchmark for land use and land cover classification.
\emph{IEEE Journal of Selected Topics in Applied Earth Observations and Remote Sensing}, 2019.
\url{https://arxiv.org/abs/1709.00029}.

\bibitem[Knight(2008)]{knight2008}
Philip A. Knight.
The Sinkhorn--Knopp Algorithm: Convergence and Applications.
\emph{SIAM Journal on Matrix Analysis and Applications}, 30(1):261--275, 2008.
\url{https://doi.org/10.1137/060659624}.

\bibitem[Krizhevsky(2009)]{krizhevsky2009}
Alex Krizhevsky. Learning Multiple Layers of Features from Tiny Images. Technical report, University of Toronto, 2009. \url{https://www.cs.toronto.edu/~kriz/cifar.html}.

\bibitem[Liang et~al.(2026)]{liang2026}
Ruiting Liang, Wanrong Zhu, and Rina Foygel Barber.
Conformal prediction after data-dependent model selection.
\emph{arXiv:2408.07066}, version 4, 2026. \url{https://arxiv.org/abs/2408.07066v4}.

\bibitem[Liu et~al.(2026)]{liu2026tail}
Shuqi Liu, Jianguo Huang, and Luke Ong.
Conformal Prediction Meets Long-tail Classification.
In \emph{Proceedings of the AAAI Conference on Artificial Intelligence}, 40(28):23828--23836, 2026.
\url{https://doi.org/10.1609/aaai.v40i28.39558}.

\bibitem[Maji et~al.(2013)]{maji2013}
Subhransu Maji, Juho Kannala, Esa Rahtu, Matthew Blaschko, and Andrea Vedaldi. Fine-Grained Visual Classification of Aircraft. \emph{arXiv:1306.5151}, 2013. \url{https://arxiv.org/abs/1306.5151}.

\bibitem[Maurer and Pontil(2009)]{maurer2009}
Andreas Maurer and Massimiliano Pontil. Empirical Bernstein bounds and sample variance penalization. In \emph{Conference on Learning Theory}, 2009. \url{https://arxiv.org/abs/0907.3740}.

\bibitem[Menon et~al.(2021)]{menon2021}
Aditya Krishna Menon, Sadeep Jayasumana, Ankit Singh Rawat, Himanshu Jain, Andreas Veit, and Sanjiv Kumar.
Long-Tail Learning via Logit Adjustment.
In \emph{International Conference on Learning Representations}, 2021.
\url{https://arxiv.org/abs/2007.07314}.

\bibitem[Ndiaye(2022)]{ndiaye2022}
Eugene Ndiaye. Stable conformal prediction sets. In \emph{International Conference on Machine Learning}, volume 162, pp. 16462--16479, 2022. \url{https://proceedings.mlr.press/v162/ndiaye22a.html}.

\bibitem[Parkhi et~al.(2012)]{parkhi2012}
Omkar M. Parkhi, Andrea Vedaldi, Andrew Zisserman, and C.~V. Jawahar.
Cats and dogs. In \emph{IEEE Conference on Computer Vision and Pattern Recognition}, 2012.
\url{https://www.robots.ox.ac.uk/~vgg/data/pets/}.

\bibitem[Radford et~al.(2021)]{radford2021}
Alec Radford et~al.
Learning transferable visual models from natural language supervision.
In \emph{International Conference on Machine Learning}, volume 139, pp. 8748--8763, 2021.
\url{https://proceedings.mlr.press/v139/radford21a.html}.

\bibitem[Romano et~al.(2020)]{romano2020}
Yaniv Romano, Matteo Sesia, and Emmanuel J. Cand\`es.
Classification with Valid and Adaptive Coverage.
In \emph{Advances in Neural Information Processing Systems}, volume 33, 2020.
\url{https://arxiv.org/abs/2006.02544}.

\bibitem[Shafer and Vovk(2008)]{shafer2008}
Glenn Shafer and Vladimir Vovk.
A tutorial on conformal prediction.
\emph{Journal of Machine Learning Research}, 9(12):371--421, 2008.
\url{https://jmlr.org/papers/v9/shafer08a.html}.

\bibitem[Silva-Rodr\'{i}guez et~al.(2025a)]{silva2025}
Julio Silva-Rodr\'{i}guez, Ismail Ben Ayed, and Jose Dolz.
Conformal Prediction for Zero-Shot Models.
In \emph{IEEE/CVF Conference on Computer Vision and Pattern Recognition}, 2025.
\url{https://arxiv.org/abs/2505.24693}.

\bibitem[Vovk(2012)]{vovk2012}
Vladimir Vovk.
Conditional validity of inductive conformal predictors.
In \emph{Asian Conference on Machine Learning}, volume 25, pp. 475--490, 2012.
\url{https://proceedings.mlr.press/v25/vovk12.html}.

\bibitem[Xie et~al.(2025)]{xie2025}
Tianmin Xie, Yanfei Zhou, Ziyi Liang, Stefano Favaro, and Matteo Sesia.
Conformal Inference for Open-Set and Imbalanced Classification.
\emph{arXiv:2510.13037}, version 1, 2025.
\url{https://arxiv.org/abs/2510.13037v1}.

\bibitem[Zeng et~al.(2025)]{zeng2025}
Hao Zeng, Kangdao Liu, Bingyi Jing, and Hongxin Wei.
Parametric Scaling Law of Tuning Bias in Conformal Prediction.
In \emph{International Conference on Machine Learning}, volume 267, pp. 74133--74156, 2025.
\url{https://proceedings.mlr.press/v267/zeng25e.html}.

\bibitem[Zhai et~al.(2023)]{zhai2023}
Xiaohua Zhai, Basil Mustafa, Alexander Kolesnikov, and Lucas Beyer.
Sigmoid Loss for Language Image Pre-Training.
In \emph{IEEE/CVF International Conference on Computer Vision}, 2023.
\url{https://arxiv.org/abs/2303.15343}.

\bibitem[Silva-Rodr\'{i}guez et~al.(2025b)]{fca2025}
Julio Silva-Rodr\'{i}guez, Leo Fillioux, Paul-Henry Courn\`ede, Maria Vakalopoulou, Stergios Christodoulidis, Ismail Ben Ayed, and Jose Dolz.
Full Conformal Adaptation of Medical Vision-Language Models.
In \emph{Information Processing in Medical Imaging}, 2025.
\url{https://arxiv.org/abs/2506.06076}.

\bibitem[Silva-Rodr\'{i}guez et~al.(2025c)]{scat2025}
Julio Silva-Rodr\'{i}guez, Ismail Ben Ayed, and Jose Dolz.
Trustworthy Few-Shot Transfer of Medical VLMs through Split Conformal Prediction.
In \emph{Medical Image Computing and Computer Assisted Intervention}, 2025.
\url{https://arxiv.org/abs/2506.17503}.

\end{thebibliography}
\end{document}